\documentclass[12pt,a4paper]{article}
\usepackage{fullpage}
\usepackage{amsmath, amsthm, amsfonts,  amssymb, latexsym}
\usepackage[colorlinks, citecolor=Red, linkcolor=Blue, urlcolor=blue]{hyperref}
\usepackage{caption}
\usepackage[numbers,sort&compress]{natbib}
\usepackage{caption}
\usepackage{authblk}
\newtheorem{thm}{Theorem}
\newtheorem{prop}[thm]{Proposition}
\newtheorem{lemma}[thm]{Lemma}
\newtheorem{cor}[thm]{Corollary}
\newtheorem{dfn}[thm]{Definition}

\newtheorem{rmk}[thm]{Remark}

\usepackage[dvipsnames]{xcolor}
\usepackage{marginnote}
\usepackage[T1]{fontenc}
\usepackage{libertine}
\usepackage{tikz}						
\usetikzlibrary{arrows,shapes,snakes,
		       automata,backgrounds,
		       petri,topaths}				
\usepackage{pst-node}
\usepackage{tikz-cd}

\newcommand{\g}{\mathfrak{g}}

\newcommand{\cR}{\mathcal{R}}

\newcommand{\bW}{\textbf{W}}
\newcommand{\eti}[1]{e_\otimes^{i{#1}/\hbar}}

\newcommand{\eps}{\varepsilon}

\newcommand{\hs}{s}

\newcommand{\wrt}{with respect to }

\newcommand{\cO}{{\mathcal{O}}}

\newcommand{\ia}{{\mathrm{int}}}

\newcommand{\bF}{{\mathbf{F}}}

\newcommand{\nn}{\nonumber}
\newcommand{\beq}{\begin{equation}}
\newcommand{\eeq}{\end{equation}}

\newcommand{\defeq}{\mathrel{:=}}

\renewcommand{\vec}[1]{{\ifnum9<1#1\mathbf{#1}\else\ifcat\noexpand#1\relax\boldsymbol{#1}\else\mathbfi{#1}\fi\fi}}

\let\oldre\Re

\renewcommand{\Re}{\oldre\mathfrak{e}\,}
\renewcommand{\ker}{\text{Ker}\,}
\renewcommand{\Im}{\text{Im}\,}

\newlength{\dhatheight}

\newcommand{\eqos}{\mathrel{\approx}}

\DeclareMathOperator{\Ker}{Ker}

\DeclareMathOperator{\supp}{supp}

\begin{document}
\title{Quantum BRST charge in gauge theories in curved space-time}
\author{Mojtaba Taslimi Tehrani}
\affil{Max-Planck Institute for Mathematics in the Sciences\\
Inselstr. 22
D-04103 Leipzig Germany\\
Institut f\"ur Theoretische Physik, Universit\"at Leipzig\\ Br\"uderstr.\ 16, 04103 Leipzig, Germany\\
 \small{Mojtaba.Taslimitehrani@mis.mpg.de}}
\maketitle
\abstract{Renormalized gauge-invariant observables in gauge theories form an algebra which is obtained as the cohomology of the derivation $[\textbf{Q}_L, -]$ with $\textbf{Q}_L$ the renormalized interacting quantum BRST charge. For a large class of gauge theories in Lorentzian globally hyperbolic space-times, we derive an identity in renormalized perturbation theory which expresses the commutator $[\textbf{Q}_L, -]$ in terms of a new nilpotent quantum BRST differential and a new quantum anti-bracket  which differ from their classical counterparts by certain quantum corrections. This identity enables us to prove different manifestations of gauge symmetry preservation at the quantum level in a model-independent fashion.}

\newpage
\tableofcontents

\newpage
\section{Introduction}

Quantum field theories with local gauge symmetry play a crucial role in our understanding of elementary particle physics by describing the type of interactions between them. The quantum aspects of such theories in flat space-time have been extensively studied.  To describe the elementary particles in the Early Universe where the curvature of space-time is not negligible, one needs to extend the framework of flat space gauge theories to the curved space setting. In \cite{Hollands:2007zg}, it was shown that the renormalized quantum Yang-Mills theories in an arbitrary, Lorentzian, globally hyperbolic curved space-time can be consistently constructed to all orders in perturbation theory. However, the proof of the statements in that reference rests on the specific form of the pure Yang-Mills interaction. The present work aims to extend this result and to investigate in a more model-independent fashion the issue of symmetry preservation at the quantum level of renormalized quantum gauge theories in curved space-times. For concreteness, we work with the pure Yang-Mills theory, however our results rest only on a certain aspect of this theory, namely the absence of ``gauge anomaly'', which is also the case in a larger class of more complicated theories with local gauge symmetry. For instance, our results remain valid for superconformal Chern-Simons matter theory in 3 dimensions \cite{Taslimi-ABJM},  a class of superconformal gauge theories in 4 dimensions \cite{deMedeiros:2013mca}, and perturbative quantum gravity \cite{barnich1995general}, \cite{Brunetti2016}.

For perturbative quantization of such theories, one necessarily has to ``fix the gauge'' which, however, breaks the gauge invariance of the underlying theory. To restore gauge invariance in the BV-BRST formalism, one enlarges the field configurations to include certain ``ghost fields''. One, furthermore, constructs a ``gauge-fixed''  and enlarged action $\hat{S}$ by requiring it to be a solution to the \emph{master equation} $(\hat{S}, \hat{S})=0$. Here, $(-,-)$ is the so-called \emph{anti-bracket} which satisfies a graded Jacobi identity. The gauge-fixed action enjoys the BV-BRST symmetry $\hat{\hs}$, which is a nilpotent derivation i.e., satisfies $\hat{\hs}^2=0$. Finally, the gauge invariant observables of the original theory are recovered as the $\hat{\hs}$-cohomology at ghost number $0$.  

In flat space-time, the quantization of such theories is conventionally performed in one of the following approaches: 
\begin{itemize}
\item \emph{The Hamiltonian approach}: One first constructs the Fock space corresponding to the (free) gauge-fixed theory which necessarily is an indefinite inner product space. One then defines \cite{Curci1976, Kugo:1979gm} the physical Hilbert space with a positive definite inner product as the cohomology of the BRST charge $\hat{Q}$ which is an operator on the Fock space. For this construction to work, and for the matrix elements of the Hamiltonian operator between physical states to be independent of the chosen gauge-fixing, the BRST charge has to be nilpotent, i.e., 
\beq
\hat{Q}^2=0.
\eeq
In this respect, it is argued that \cite{henneaux1992quantization} ``the nilpotency of $\hat{Q}$ is the quantum expression of the gauge invariance''. Usually, one has to check the nilpotency of $\hat{Q}$ in a case-dependent procedure, and this turns out to be a highly non-trivial task involving regularization and renormalization of the composite operator $\hat{Q}^2$. 

\item \emph{The functional integral approach}: One defines the quantized gauge theory in terms of an effective action $\Gamma = \hat{S} + O(\hbar)$ which is the generating functional of one-particle irreducible Feynman diagrams. Gauge invariance at the quantum level is then expressed by the ``Slavnov-Taylor identity'' in the ``Zinn-Justin'' form \cite{Zinn-Justin} 
\beq\label{Zinn-Justin-eq}
(\Gamma, \Gamma)=0,
\eeq
which for $\hbar=0$ is reduced to the master equation. The graded Jacobi identity of the anti-bracket, in turn, implies that the potential obstruction to \eqref{Zinn-Justin-eq} (the ``gauge anomaly'') satisfies a consistency condition of a cohomological nature, namely that (the leading $\hbar$-order coefficient of) the gauge anomaly belongs to the cohomology ring $H_1(\hat{\hs}|d,M)$ of $\hat{\hs}$ modulo $d$ at ghost number $1$. For theories in which this cohomology ring is trivial, \eqref{Zinn-Justin-eq} is shown to be fulfilled by finite renormalization order by order in $\hbar$.  

\end{itemize}

Contrary to the flat space-time setting where the quantum fields can be represented as operators on a preferred Hilbert space containing the unique Poincar\'{e}-invariant vacuum state, in a generic (globally hyperbolic) curved space-time there is no preferred vacuum state and hence no canonical Hilbert space representation of the theory. 
We therefore employ the framework of locally covariant quantum field theory \cite{Brunetti:1995rf}\cite{Brunetti:1999jn},\cite{Hollands:2001nf}, \cite{Hollands:2001fb}, \cite{Hollands:2002ux} (see \cite{Hollands:2014eia}  for a recent review) to study such theories. In this framework, one formulates the QFT coherently on \emph{all} space-times and views the renormalized  \emph{interacting quantum fields} $\mathcal{O}_L$, under interaction $L$, as elements of an abstract algebra, which can be constructed in perturbation theory. Within this algebra, one defines the algebra of physical, gauge invariant observables as the cohomology of the derivation $[Q_L, -]$ generated by the renormalized quantum BRST charge $Q_L$. Upon a choice of a physical state \footnote{This non-canonical choice of representation is related to the physical questions one wants to study and is not discussed here.} and under certain technical conditions on the background space-time, this cohomology algebra turns out to admit a positive definite Hilbert space representation if $Q_L^2=0$ \cite{dutsch1999local, Hollands:2007zg}. 

Our main result in the present work is proving that if the cohomology ring $H_1(\hat{\hs}|d,M)$ is trivial, then $Q_L^2=0$, and thus relating the above two criterion of gauge invariance at the quantum level. 
We prove this by first showing that under this cohomological condition, the derivation $[Q_L, -]$ is nilpotent (Theorem \ref{QL2=0}), and hence the algebra of gauge invariant observables can indeed be defined as the cohomology of this derivation. Then, it follows from the graded Jacobi identity of the commutator that $Q_L^2 = \frac{1}{2}[Q_L, Q_L]$ vanishes. Thus, the problem of proving the nilpotency of the renormalized charge is reduced to an algebraic problem of a cohomological nature. Our proof is a significant improvement in the state of affairs over \cite{Hollands:2007zg}. There, the nilpotency of the charge was shown to hold based on a case-dependent proof which requires, in addition to the triviality of $H_1(\hat{\hs}|d,M)$, the precise form of the current of pure Yang-Mills theory and certain identities derived from it, as well as the triviality of a higher cohomology class which seem to hold only for this specific theory. 

The key identity that we derive in this part of the work, which forms the basis of the proof of our main result, is called the \emph{interacting anomalous Ward identity} (Theorem \ref{[QL, OL]}). It is a master identity for the commutator of the quantum BRST charge $Q_L$ and the generating functional of the renormalized interacting time-ordered products $T_{L,n}(\mathcal{O}_1 \otimes \dots  \otimes \mathcal{O}_n)$. Here $\mathcal{O}_1, \dots, \mathcal{O}_n$ are local fields and $T_{L,1}(\mathcal{O})= \mathcal{O}_L$. When evaluated in a state, such expressions give the renormalized time-ordered correlation functions of the theory. For one local field $\cO$, this identity will give
\begin{align} \label{[QL,O]}
[Q_L, \mathcal{O}_L] &= i \hbar (\hat{q} \mathcal{O})_L,
\end{align}
where $\hat{q} \mathcal{O} = \hat{\hs} \mathcal{O} +O(\hbar)$ is called the \emph{quantum BRST operator} \eqref{hat-q}
 which is nilpotent, i.e., $\hat{q}^2=0$. 
For two local fields $\cO_1, \cO_2$, we will obtain
\begin{align}\nn
[Q_L, T_{L,2}(\mathcal{O}_1\otimes \mathcal{O}_2)] = i \hbar T_{L,2}\big( \hat{q} \mathcal{O}_1 \otimes \mathcal{O}_2 + (-1)^{\eps_1} \mathcal{O}_1\otimes \hat{q} \mathcal{O}_2 \big) + (-1)^{\eps_1} \hbar^2 \big({(\mathcal{O}_1, \mathcal{O}_2)}_\hbar \big)_L,\\ \label{QL,O1O2}
\end{align}
where $\eps_1$ is the Grassmann parity of $\cO_1$, and where $(\mathcal{O}_1,\mathcal{O}_2)_\hbar = (\mathcal{O}_1,\mathcal{O}_2) + O(\hbar)$ is called the \emph{quantum anti-bracket} \eqref{q-anti-bracket} which measures the failure of $\hat{q}$ to be a derivation. It is compatible with $\hat{q}$, in the sense that $\hat{q} {( \mathcal{O}_1 , \mathcal{O}_2 )}_\hbar = {(\hat{q}\mathcal{O}_1, \mathcal{O}_2 )}_\hbar - (-1)^{\eps_1} {(\mathcal{O}_1, \hat{q}\mathcal{O}_2 )}_\hbar$, and satisfies a \emph{quantum Jacobi identity} \eqref{An-Jacobi-ferm}. 
\subsubsection*{Notations}

In the body of the paper, we encounter local fields $\mathcal{O} = \mathcal{O}_0 + \lambda \mathcal{O}_1 + \lambda^2 \mathcal{O} + \dots$ which are $p$-forms expanded into powers of the coupling constant $\lambda$. For the integrated operators we use $F= \int  \mathcal{O}_0 + f\lambda \mathcal{O}_1 + f^2 \lambda^2 \mathcal{O} + \dots$ where $f \in \Omega_0^{4-p}(M)$ is an IR cutoff which is equal $1$ in some region of interest. We symbolically write all such expressions as $F= \int f \mathcal{O}$. For the particular case of $\textbf{L}_{\text{int}}$, the interaction Lagrangian, we denote the ``cutoff interaction'' with $L= \int f \textbf{L}_{\text{int}}$, and the ``true interaction'' with $I = \int \textbf{L}_{\text{int}}$. 
We always write the interacting BRST charge with the cutoff interaction $\textbf{Q}_L$, and avoid using $\textbf{Q}_I$ (which can be defined as the algebraic adiabatic limit of $\textbf{Q}_L$ at the end of calculations). 
Moreover, in many places, we write $\mathcal{O}_1 \otimes \dots \otimes \mathcal{O}_n$ as a short form for $\mathcal{O}_1(x_1) \otimes \dots \otimes \mathcal{O}_n(x_n)$.
Finally, we express all the generating functional identities only for bosonic functionals and explain in appendix \ref{app1} how the correct signs for fields with arbitrary Grassmann parity in such identities can be obtained. 
\section{Classical gauge theory}\label{Classical}
We take for definiteness the example of the pure Yang-Mills theory on a globally hyperbolic space-time $(M,g)$ worked out in \cite{Hollands:2007zg}. It turns out that, the results of this work can be generalized to all theories with local gauge symmetry, such as superconformal Yang-Mills theory \cite{deMedeiros:2013mca}, and superconformal Chern-Simons-matter theory \cite{Taslimi-ABJM} in which a certain cohomology class is trivial. Here, we briefly review the setting for the classical theory. 

The classical Yang-Mills theory with gauge group $G$ is the dynamical theory of a $G$-gauge connection $\mathcal{D}= \nabla + i\lambda A$, where $\nabla$ is the Levi-Civita connection on $(M,g)$ and $A$ is a $\mathfrak{g}$-valued one form. The action functional is given by
\begin{equation}
S_{\text{YM}} = - \frac{1}{2} \int_M \text{tr} (F \wedge *F),
\end{equation}
where $F$ is the curvature of $\mathcal{D}$.
For the purpose of perturbative quantization, the equations of motion for $A$ generated by the action have to be of hyperbolic type. However, this is not the case for $S_{\text{YM}}$; one has to fix the gauge in order to render the free field equations hyperbolic. The resulting gauge-fixed theory enjoys the BRST symmetry $\hat{s}$, if one augments the field content of the theory by further dynamical fields (ghosts) and non-dynamical fields (anti-fields), as introduced below.

Let $\{T_I\}$, $I= 1, 2, \dots, \text{dim}\mathfrak{g}$ be a basis for the Lie algebra $\mathfrak{g}$ of the Lie group $G$. Relative to this basis, we have $A= A^I T_I = A^I_\mu T_I dx^\mu$. 
Let us denote the set of all dynamical fields by $\Phi=(A^I, C^I, \bar{C}^I, B^I)$, where $C, \bar{C}$ are called ghosts and $B$ is an auxiliary field with algebraic equations of motion, and their corresponding anti-fields by $\Phi^\ddag = (A_I^{\ddag}, C_I^{\ddag}, \bar{C}_I^{\ddag}, B_I^\ddag)$. The action of the BRST differential $\hat{s}$ on all fields is given by:
\begin{align}
& \hat{s} A_\mu^I = D_\mu C^I, \hspace{4 mm} \hat{s} C^I = - \frac{i \lambda}{2} {f^I}_{JK} C^J C^K,\hspace{4 mm} \hat{s} \bar{C}^I  = B^I,\hspace{4 mm} \hat{s} B^I  = 0.
\end{align}
One can now assign a ``ghost number'' to all the above fields which is given in table \ref{table1}. The ghost number defines a grading on the space of all fields, and the BRST differential increases the ghost number by one unit while leaves the dimension unchanged.
\begin{table}[ht]
\begin{center}
    \begin{tabular}{   | l | l | l | l | l | l | l | l | p{0.4 cm}  | }
    \hline
    \textit{Fields} & $A^I$ &  $C^I$ & $\bar{C}^I$ &$B^I$ & $A_I^{\ddag}$ & $C_I^{\ddag}$ & $\bar{C}_I^{\ddag}$ & $B_I^\ddag$ \\ \hline \hline 
    Dimension & 1 & 0 &   $ 2 $& 2 & 3&4 &2 &2 \\ \hline
   Ghost number & 0 & 1 &   -1 & 0 & -1 & -2 & 0 & -1\\ \hline
   Grassmann parity & 0 & 1 &   1 & 0 &  1& 0 & 0 & 1 \\ \hline
    \end{tabular}
    \caption{Basic fields and their data.} \label{table1}
  \end{center}
    \end{table}
To define how $\hat{s}$ acts on the anti-fields, consider the following extended action
\begin{equation}
\hat{S} = S_{\text{YM}} + \hat{s} \psi - \int \hat{s} \Phi \cdot \Phi^\ddag,
\end{equation} 
where $\psi$ is the ``gauge-fixing fermion''. It is chosen in such a way that $\hat{S}$ gives rise to hyperbolic field equations for all fields $\Phi$. A conventional choice for $\psi$ is $\psi= \int_M \bar{C}_I(\nabla^\mu A_\mu^I + \frac{1}{2} B^I)$ which implements the Feynman gauge.
Now, for any observable $\mathcal{O}$ we define
\begin{equation}
\hat{s} \mathcal{O} = (\hat{S}, \mathcal{O}),
\end{equation}
where $(-,-)$ is the so-called \emph{anti-bracket} defined by
\begin{equation}\label{antibracket}
(\mathcal{O}_1, \mathcal{O}_2) := \int_M \frac{\delta_R \mathcal{O}_1}{\delta \Phi(x)}  \frac{\delta_L \mathcal{O}_2}{\delta \Phi^\ddag(x)} -  \frac{\delta_R \mathcal{O}_1}{\delta \Phi^\ddag(x)} \frac{\delta_L \mathcal{O}_2}{\delta \Phi(x)},
\end{equation}
cf. \cite{Rejzner:2011au} for a definition of left and right derivatives \wrt fields with Grassmann parity. The anti-bracket has the following symmetry property
\begin{equation}\label{graded-symm}
( \mathcal{O}_1 , \mathcal{O}_2 ) = (-1)^{(\eps_1 +1)( \eps_2+1)+1} ( \mathcal{O}_2 , \mathcal{O}_1 ).
\end{equation}
and satisfies the graded Jacobi identity
\begin{align}\nonumber
(-1)^{(\eps_1+1) (\eps_3+1)}  (\mathcal{O}_1, ( \mathcal{O}_2 , \mathcal{O}_3 )) &+ (-1)^{(\eps_2+1) (\eps_1+1)}  (\mathcal{O}_2, ( \mathcal{O}_3 , \mathcal{O}_1 ))\\ \label{class-Jacobi}
 &+(-1)^{(\eps_3+1) (\eps_2+1)}  (\mathcal{O}_3, ( \mathcal{O}_1 , \mathcal{O}_2 ))  =0.
\end{align}
From \eqref{graded-symm} and \eqref{class-Jacobi}, together with $(S, \mathcal{O}) =\hat{s} \mathcal{O}$, it follows that
\begin{align}\label{s-antibracket}
\hat{s} ( \mathcal{O}_1 , \mathcal{O}_2 ) =( \hat{s} \mathcal{O}_1 , \mathcal{O}_2 ) - (-1)^{\eps_1} ( \mathcal{O}_1 , \hat{s} \mathcal{O}_2 ). 
\end{align}
Note that in particular, it follows that $\hat{s} \Phi^\ddag(x) = \delta_R \hat{S}/ \delta \Phi(x)$ and $(\Phi_i(x), \Phi^{\ddag j}(y)) = {\delta_i}^j \delta(x,y)$. Moreover, from the definition of $\hat{S}$ we have
\begin{equation}
\hat{s} \hat{S} = (\hat{S}, \hat{S})= \hat{s}^2 = 0.
\end{equation}
\subsubsection*{Local and covariant fields}
\begin{dfn}\label{local-covariant}
Let $\mathcal{C}$ be the space of enlarged field configurations $(\Phi, \Phi^\ddag)$ together with the metric $g$.
\begin{enumerate}
\item Let $f : M' \rightarrow M$ be an isometric embedding which preserves the causal structures. 
A \textbf{local-covariant} functional $\mathcal{O}$ on $\mathcal{C}$ satisfies
\begin{equation}
f^* \mathcal{O} [g, \Phi, \Phi^\ddag] = \mathcal{O}[f^*g, f^* \Phi, f^* \Phi^\ddag].
\end{equation}

\item $\textbf{P}(M)= \oplus_{p,q} \textbf{P}^p_q(M)$, where each $\textbf{P}^p_q(M)$ is defined to be the space of all $\wedge^p(TM)$-valued polynomial, local and covariant functionals with ghost number $q$.
\end{enumerate}
\end{dfn}

There are two important theorems regarding the nature of $\textbf{P}^p_q(M)$. 
First, the \textit{Thomas replacement theorem} \cite{Iyer:1994ys}, which states that the dependence of every element $\mathcal{O} \in \textbf{P}^p_q(M)$ on the metric and, at each point $x\in M$, on $\Phi(x), \Phi^\ddag(x)$  is of the form
\begin{align}
\mathcal{O}&= \mathcal{O}\big( g_{\mu \nu}(x), {R^\mu}_{\nu \rho \sigma}(x), \dots, \nabla_{(\mu_1 \dots} \nabla_{\mu_k)} {R^\mu}_{\nu \rho \sigma}|_x, \nabla_{(\mu_1 \dots} \nabla_{\mu_k)} \Phi |_{x}, \nabla_{(\mu_1 \dots} \nabla_{\mu_k)}\Phi^\ddag|_x \big)
\end{align}
where ${R^\mu}_{\nu \rho \sigma}$ is the Riemann tensor. Therefore, if we assign dimension $1$ to $\nabla_\mu$, we can assign a dimension to all elements of $\textbf{P}^p_q(M)$.
Second, the \textit{algebraic Poincare lemma} \cite{wald1990identically} states that if for some $\mathcal{O} \in \textbf{P}^p_q(M)$, $d \mathcal{O}=0$, then there exists another $\mathcal{O}' \in \textbf{P}^{p-1}_q(M)$ such that $\mathcal{O}= d \mathcal{O}'$. Note that this is a property of $d$-cohomology for functionals of $\Phi, \Phi^\ddag$, and holds even for space-times with non-trivial de Rham cohomology.

The $q$-th cohomology ring of $\hat{s}$ at form degree $p$ is defined by
\begin{equation}
H_q^p (\hat{s}, M):= \frac{\{ \ker \hat{s}: \textbf{P}^p_q(M) \rightarrow \textbf{P}^p_{q+1}(M)\}}{\{\Im \hat{s}: \textbf{P}^p_{q-1}(M) \rightarrow \textbf{P}^p_{q}(M) \}}.
\end{equation}
We will show in section \ref{WardIdentities}, that the anomaly $A = \int_M a(x)$ is a formal power series in $\hbar$ whose leading order contribution $A^m$ is an element of $H_1^4 (\hat{s}, M)$. Equivalently, the local function $a^m(x)$ belongs to the \textit{cohomology rings of $\hat{s}$ modulo $d$} defined by
\begin{equation}
H_q^p (\hat{s}|d, M):= \frac{\{ \mathcal{O}_q^p | \hat{s} \mathcal{O}^p_q= d \mathcal{O}_{q+1}^{p-1} \}}{\{  \mathcal{O}^p_q | \mathcal{O}^p_q= \hat{s} \mathcal{O}^p_{q-1} +  d \mathcal{O}_{q}^{p-1} \}}.
\end{equation}

\subsubsection*{Gauge-invariant observables and the Noether current}

The local and covariant functionals introduced above, of course contain all possible gauge-variant functionals of the enlarged (un-physical) theory. It turns out \cite{Barnich:2000zw} that one can recover the gauge-invariant observables of the original, physical theory as the following cohomology
\begin{equation}
\{ \text{classical gauge-invariant observables} \} = H_0(\hat{s}, M).
 \end{equation} 
According to the Noether's theorem, the invariance of $\hat{S}$ under $\hat{s}$ results in the existence of a current $\textbf{J}(x) \in \textbf{P}^3_1(M)$ (the BRST current) which is conserved $d \textbf{J}=0$ once the equations of motion hold. Indeed, we have
\begin{equation}\label{dJ(x)}
d \textbf{J} (x) =  (\hat{S}, \Phi^\ddag(x))(\Phi(x), \hat{S}).
\end{equation} 
Let us assume that $(M,g)$ contains a compact Cauchy surface $\Sigma$. Then, there exists a corresponding BRST charge $\textbf{Q}$, defined by
\begin{equation}\label{BRSTcharge}
\textbf{Q}= \int_M \gamma \wedge \textbf{J},
\end{equation} 
where $\gamma$ is a closed 1-form on $M$ with compact support such that $\int_M \gamma \wedge \alpha= \int_\Sigma \alpha$ for any closed 3-form $\alpha$. 

\section{Quantum gauge theory in curved space-time}\label{QFTCST}

We now turn to the quantization of the classical field theory introduced in the previous part. 
To this end, we employ the ideas of causal perturbation theory \cite{epstein1973role} adopted to the framework of locally covariant field theory \cite{Brunetti:2001dx},\cite{Hollands:2001nf}, \cite{Hollands:2001fb} which aims to construct the algebra of observables of the theory.  However, in our case we begin with constructing $\hat{\textbf{W}}_L$ which is the  quantization of the enlarged theory including gauge-variant and non-observable elements as a perturbation in $\lambda$ around the algebra of free theory $\hat{\textbf{W}}_0$. 
\subsection{Free quantum theory and renormalization schemes}\label{FreeQFT}
Let us begin with reviewing the construction of the algebra of free quantum fields $\hat{\textbf{W}}_0$ corresponding to the enlarged theory defined by $\hat{S}_0$. Here we split the extended action
\begin{equation}\label{hatS0+I}
\hat{S} = \hat{S}_0 + \lambda \hat{S}_1 + \lambda^2 \hat{S}_2 \equiv \hat{S}_0 + I,
\end{equation} 
where $\hat{S}_0$ is quadratic in all fields and anti-fields, and where $I = \int_M \textbf{L}_{\text{int}}$ with $\textbf{L}_{\text{int}} = \lambda \textbf{L}_1 + \lambda^2 \textbf{L}_2 \in \textbf{P}_0^4(M)$. Let us also denote $S_0 =  {\hat{S}_0}|_{\Phi^\ddag=0} = S_{\text{YM},0} + \hat{s}_0 \psi$.
 
Consider now the free differential operator $P_{ij}^0$, with $P_{ij}^0 \Phi^{j}(x)={\delta {S}_0}/{\delta \Phi^i(x)}$, and 
\begin{equation}
P^0_{ij} = \begin{pmatrix}
 g^{\mu \nu} \Box  - \nabla^\mu \nabla^\nu - R^{\mu \nu}& 0 & 0 & \nabla^\mu \\
0 & 0 & \Box   & 0  \\
   0& -\Box & 0 &  0 \\
 \nabla^\mu & 0 & 0 & - 1 \\
   \end{pmatrix}.
\end{equation}
An important ingredient in constructing $\hat{\textbf{W}}_0$ is an arbitrary but fixed 2-point function of Hadamard type $\omega^{ij}(x,y)$. It is a distribution on $M \times M$, which satisfies
\begin{itemize}
\item[(1)] $(P_{ik}^0 \otimes \textbf{1}) \omega^{kj}(x,y) =0 = (\textbf{1} \otimes P_{ik}^0) \omega^{kj}(x,y)$, 
\item[(2)]  $ \omega^{ij}(x,y)-\omega^{ij}(y,x)= i \Delta^{ij}(x,y)$, where $\Delta^{ij}(x,y)$ is the causal propagator of $P^0_{ij}$,
\item[(3)] a specific wave-front set bound (see \cite{Radzikowski:1996pa}).
\end{itemize}
 Explicitly,
\begin{equation}\label{Hadamard-2point}
\omega^{ij}(x,y) = k_{IJ} \otimes \begin{pmatrix}
  \omega^{\mu \nu}(x,y) & 0 & 0 & -i \nabla_\nu \omega^{\mu \nu}(x,y)  \\
  0 &  0 & i \omega(x,y)  & 0   \\
  0 & - i \omega(x,y) & 0& 0 \\
 -i \nabla_\mu \omega^{\mu \nu}(x,y) &0 & 0   & 0\\
 \end{pmatrix},
\end{equation}
where $\omega(x,y)$, $\omega^{\mu \nu}(x,y)$ are scalar and vector two point functions respectively (see e.g. \cite{dewitt1960radiation}). They satisfy the following consistency relations
\begin{equation}\label{omega-consistency}
\nabla_\mu \omega^{\mu \nu}(x,y)= -\nabla^\nu \omega(x,y), \hspace{5 mm} \nabla_\nu \omega^{\mu \nu}(x,y) = - \nabla^\mu \omega(x,y).
\end{equation}

\begin{dfn}[\cite{Hollands:2001nf}, \cite{Hollands:2007zg}] 
\item (1)
The \emph{off-shell free algebra} $\hat{\textbf{W}}_0(M,g)$ is the *-algebra generated by the identity \textbf{1} and elements
\begin{align} \nonumber
F(u) = \int & u_{i_1 \dots i_n}^{k_1 \dots k_m} (x_1, \dots, x_n ; y_1, \dots, y_m) : \Phi^{i_1}(x_1) \dots \Phi^{i_n}(x_n):_\omega \Phi^\ddag_{k_1}(y_1) \dots \Phi^\ddag_{k_m}(y_m)\\
& d x_1 \dots d x_n d y_1 \dots d y_m
\end{align}
In this expression,
\begin{align}\nonumber
&{:\Phi(x_1)\dots \Phi(x_n):}_\omega\\
 & = \frac{\delta^n}{i^n \delta f(x_1) \dots \delta f(x_n)} \exp_\star {\left( i \int_M f(x) \Phi(x) + \frac{\hbar}{2} \int_{M^2} \omega(x,y) f(x) f(y)\right)\Big|}_{f=0},
\end{align}
with the star product of two basic fields being defined by
\begin{equation}\label{starprod}
\Phi^i(x) \star \Phi^j(y)  = \Phi^i(x) \cdot \Phi^j(y)  + \hbar \omega^{ij}(x,y),
\end{equation}
and $u$ is a distribution subject to the following wave front set condition in the variables $x_1, \dots, x_n$
\begin{equation}
WF(u) \cap \bigcup_{x \in M} [(\bar{V}_x^+)^{\times n} \cup (\bar{V}_x^-)^{\times n}]= \emptyset,
\end{equation}
where $\bar{V}^{\pm}_x$ is the closure of the future/past light cone at $x \in M$,
but $u$ is not subject to any wave front set condition in the variables $y_1, \dots, y_m$. The *-operation, denoted by $\dag$, is defined by $F(f)^\dag = F(\bar{f})$.

\item (2) The \emph{on-shell free algebra} $\hat{\bF}_0(M,g)$  is the quotient 
\begin{equation}\label{hat-F-0}
\hat{\bF}_0  = \hat{\textbf{W}}_0/ \mathcal{J}_0,
\end{equation}
where $\mathcal{J}_0$ is the $\star$-ideal \emph{generated by the equations of motion}:
\begin{align} \nn 
&\mathcal{J}_0= \big\{ \int u_{i_1 \dots i_n}^{k_1 \dots k_m} (x_1, \dots, x_n ; y_1, \dots, y_m)\\ \label{eom-ideal}
& \quad : \Phi^{i_1}(x_1) \dots \frac{\delta \hat{S}_0}{\delta \Phi^{i_i}(x_i)}\dots \Phi^{i_n}(x_n):_\omega \Phi^\ddag_{k_1}(y_1) \dots \Phi^\ddag_{k_m}(y_m) d x_1 \dots d x_n d y_1 \dots d y_m   \big\}.
\end{align}
\end{dfn}

Therefore in $\hat{\bF}_0$, the basic dynamical fields $\Phi^i(x)$ satisfy the free field equations $P_{ij}^0 \Phi^{j}(x) =J_i(x)$, where $J_i=(0, \nabla^\mu A^\ddag_\mu, 0, C^\ddag)$ is an  anti-field dependent source term.

Note that \eqref{starprod} implies
\begin{align}\label{Comm-rel-1}
[\Phi^i(x), \Phi^j(y)] &= i \hbar \Delta^{ij}(x,y) \textbf{1},
\end{align}
where $[\Phi^i(x), \Phi^j(y)] = \Phi^i(x) \star \Phi^j(y) - (-1)^{\eps_i \eps_j} \Phi^j(y) \star \Phi^i(x)$ is the graded commutator, satisfying
\begin{align}
 [\mathcal{O}_1 , \mathcal{O}_2 ] = - (-1)^{\eps_1 \eps_2}  [\mathcal{O}_2 , \mathcal{O}_1 ],
\end{align}
and satisfying the graded Jacobi identity
\begin{align}\label{Jacobi-commutator}
(-1)^{\eps_1 \eps_3} [\mathcal{O}_1 ,[\mathcal{O}_2 , \mathcal{O}_3 ]  ] +(-1)^{\eps_3 \eps_2} [\mathcal{O}_3 ,[\mathcal{O}_1 , \mathcal{O}_2 ]  ]  + (-1)^{\eps_1 \eps_2} [\mathcal{O}_2 ,[\mathcal{O}_3 , \mathcal{O}_1 ]  ]  =0. 
\end{align}


\subsubsection*{Renormalization schemes and finite counter terms}\label{timeorderedproducts}

In the algebraic formulation of QFT \`{a} la causal perturbation theory \cite{epstein1973role}, one directly formulates the renormalized theory in terms of \emph{time-ordered products} or \emph{renormalization schemes} $T$ which are defined to satisfy a set of physically reasonable renormalization conditions. These quantities are defined in the off-shell algebra $\hat{\textbf{W}}_0$ which turns out to be more suitable for perturbation theory than $\hat{\bF}_0$. In the next step, one constructs the algebra of interacting quantum fields $\hat{\textbf{W}}_L= \hat{\textbf{W}}_0 [[\lambda]]$ as formal power series in $\lambda$ with coefficients in $\hat{\textbf{W}}_0$ which is reviewed in the next section \ref{interactingtheory}. 

\begin{dfn}[\textbf{Renormalization schemes or time-ordered products}]\label{ren-schemes}
A renormalization scheme $T$ is a collection of multi-linear maps
\begin{equation}
T_n : \textbf{P}(M)^{\otimes n} \rightarrow D'(M^n; \hat{\textbf{W}}_0),
\end{equation} 
that is, each $T_n(\mathcal{O}_1(x_1)\otimes \dots \otimes \mathcal{O}_n(x_n))$ is a $\hat{\textbf{W}}_0$-valued distribution in $n$ space-time variables $x_1, \dots, x_n$. It satisfies the following axioms (renormalization conditions)
\begin{enumerate}
\item[\textbf{T1})] \textbf{Locality and covariance.} For locally isometric space-times $(M, g)$ and $(M', g')$, it holds
\begin{equation}
\alpha_\psi \circ T_g = T_{g'} \circ \otimes \psi_*.
\end{equation}
Here, $\psi : M \rightarrow M'$ is a causality preserving isometric embedding, i.e. $\psi^{*} g' = g$, and $\alpha_\psi$ is the corresponding canonical homomorphism
\begin{equation}\label{alphapsi}
\alpha_\psi: \hat{\textbf{W}}_0(M, g) \rightarrow \hat{\textbf{W}}_0(M', g'), \hspace{5 mm} \alpha_\psi(F(u)) = F(\psi_* u),
\end{equation}
with $\psi_* : \textbf{P}(M) \rightarrow \textbf{P}(M')$ being the natural push-forward map. 
\item[\textbf{T2})] \textbf{Scaling.} Each $T_n$ has a poly-homogeneous scaling behavior under $g \mapsto \mu^2 g$, cf. \cite{Hollands:2001fb} 

\item[\textbf{T3})] \textbf{Microlocal spectrum condition.} The wave-front set of each $T_n(\mathcal{O}_1(x_1)\otimes \dots \otimes \mathcal{O}_n(x_n))$ is bounded by a specific subset of $T*M^n$, cf. \cite{Hollands:2001fb}.
\item[\textbf{T4})] \textbf{Smoothness and \textbf{Analyticity
\footnote{It is shown in \cite{Khavkine2016} that the analyticity assumption may be dropped.}
.}} Each $T_n$ is a smooth and analytic functional of the metric $g$.
\item[\textbf{T5})] \textbf{Graded symmetry.} Each $T_n$ is graded symmetric under a permutation of its arguments, that is,
\beq
T_{n}(\dots \otimes \cO_i \otimes \cO_{j} \otimes \dots ) = (-1)^{\eps_i \eps_{j}}T_{n}(\dots \otimes \cO_{j} \otimes \cO_{i} \otimes \dots ).
\eeq
\item[\textbf{T6})] \textbf{Unitarity.} Renormalization schemes are unitary in the following sense
\begin{equation}
[T_n(\otimes_i \mathcal{O}_i(x_i)^*)]^\dag = \sum_{I_1 \sqcup \dots \sqcup I_j = \underline{n}} (-1)^{n+j} T_{|I_1|} (\otimes_{i\in I_1} \mathcal{O}_i(x_i)) \star \dots \star T_{|I_j|} (\otimes_{i\in I_j} \mathcal{O}_j(x_j)), 
\end{equation}
where $I_1, \dots , I_j$ are pairwise disjoint subsets of $\underline{n}= \{ 1, \dots, n\}$.
\item[\textbf{T7})] \textbf{Causal factorization.} For ${x_1, \dots, x_i} \cap J^{-}(\{ x_{i+1}, \dots, x_n\}) = \emptyset$, it holds
\begin{align} \nonumber
& T_n(\mathcal{O}_1(x_1)\otimes \dots \otimes \mathcal{O}_n(x_n))\\  \label{CFaxiom}
& = T_i(\mathcal{O}_1(x_1)\otimes \dots \otimes \mathcal{O}_i(x_i)) \star T_{n-i}(\mathcal{O}_{i+1}(x_{i+1}) \otimes \dots \otimes \mathcal{O}_n(x_n)).
\end{align}
\item[\textbf{T8})] \textbf{Commutator.} The commutator of each $T_n$ with a basic field $\Phi(x)$ is implemented as
\begin{align}\nonumber
&[T_n(\mathcal{O}_1(x_1)\otimes \dots \otimes \mathcal{O}_n(x_n)), \Phi^i(x)] \\
&=  i \hbar \sum_{k=1}^n  T_n \big(\mathcal{O}_1(x_1)\otimes \dots \otimes \int_M \Delta^{ij} (x,y) \frac{\delta \mathcal{O}_k(x_i)}{\delta \Phi^j(y)} \otimes \dots \otimes \mathcal{O}_n(x_n) \big).	
\end{align}
\item[\textbf{T9})] \textbf{Free field equation} The free field equations, $\frac{\delta S_0}{\delta \Phi(x)}=0$, is implemented in the following sense
\begin{equation}
T_{n+1}(\frac{\delta S_0}{\delta \Phi(x)} \otimes \mathcal{O}_1\otimes \dots \otimes \mathcal{O}_n) \eqos i \hbar \sum_{i=1}^n  T_n(\mathcal{O}_1\otimes \dots \otimes \frac{\delta \mathcal{O}_i(x_i)}{\delta \Phi(x)}\otimes \dots \otimes \mathcal{O}_n),	
\end{equation}
where $\eqos$ means equal modulo the ideal $\mathcal{J}_0$ of free equations of motion \eqref{eom-ideal}
\beq
\label{eq:approx_0}
F \eqos G \quad \Leftrightarrow \quad F - G \in \mathcal{J}_0.
\eeq
\item[\textbf{T10})] \textbf{Action Ward identity}
\footnote{Note that the axiom \textbf{T10} implies $T_n(\dots \otimes \int f \wedge d \mathcal{O} \otimes \dots) = T_n(\dots \otimes -\int df \wedge  \mathcal{O} \otimes \dots)$, which in turn means that each $T_n$ may be equivalently viewed as a map on the space of local action functionals.}\label{Fnote2}
 $T_n$ commutes with derivatives, i.e.
\begin{equation}\label{ActionWI}
d_{x_i}T_n(\mathcal{O}_1(x_1)\otimes \dots \otimes \mathcal{O}_n(x_n)) = T_n(\mathcal{O}_1(x_1) \otimes \dots \otimes d_{x_i} \mathcal{O}_i(x_i) \otimes \dots \otimes \mathcal{O}_n(x_i)).
\end{equation}
\end{enumerate}
\end{dfn}
The crucial fact about the renormalization schemes, proved in \cite{Hollands:2001fb}, is that they exists and are unique up to a well-characterized, local and covariant ``renormalization ambiguity''. This existence and uniqueness theorem is precisely formulated in the following.
\begin{thm}[\textbf{The main theorem of renormalization theory} \cite{Hollands:2001fb}, \cite{Hollands:2001nf}] Renormalization schemes satisfying the axioms of definition \ref{ren-schemes} exist. Let $T$ and $\tilde{T}$ be two renormalization schemes which satisfy those axioms. Then they are related via
\begin{align}\label{tildeT-T}
\tilde{T}_n(\mathcal{O}_1 \otimes \dots \otimes \mathcal{O}_n) &=  \sum_{I_0 \cup \dots \cup I_r \subset \underline{n}} T_{r +1} \big( \bigotimes_{ k} \big(\frac{\hbar}{i}\big)^{|I_k|} D_{|I_k|}(\bigotimes_{i \in I_k} \mathcal{O}_i) \otimes \bigotimes_{j \in I_0}  \mathcal{O}_j \big) ,
\end{align}
where the sum runs over all partitions $I_0 \cup \dots \cup I_r$  of the set $\underline{n}=\{1, \dots, n\}$ into pairwise disjoint non-empty subsets, and where $D= (D_n)_{n \ge 1}$ is a hierarchy of maps 
\begin{equation}
D_n: \textbf{P}(M)^{\otimes n} \rightarrow \textbf{P}^{k_1/ \dots /k_n}(M)[[\hbar]] ,
\end{equation}
where $\textbf{P}^{k_1/ \dots /k_n}(M)[[\hbar]]$ is the space of all distributional local and covariant functionals supported on the total diagonal, and are a $k_i$-form in the $i$-th argument $x_i$, for all $i= 1, \dots, n$, and satisfy
\begin{enumerate}
\item [\textbf{D1})] $D_n (\mathcal{O}_1(x)\otimes \dots \otimes \mathcal{O}_n(x))$ is of order $O(\hbar)$ if all $\mathcal{O}_i$ are of order $O(\hbar^0)$,
\item [\textbf{D2})]  Each $D_n$ is locally, and covariantly constructed out of $g$, and is an analytic functional of $g$, 
\item [\textbf{D3})]  Each $D_n (\mathcal{O}_1(x)\otimes \dots \otimes \mathcal{O}_n(x))$ is supported on the total diagonal 
\begin{equation}
\Delta_n = \{(x,x, \dots, x) | x \in M\} \subset M^n,
\end{equation}
\item [\textbf{D4})]  Each $D_n$ is graded symmetric,
\item [\textbf{D5})]  The maps $D_n$ are real $D_n (\mathcal{O}_1(x)\otimes \dots \otimes \mathcal{O}_n(x))^*= D_n (\mathcal{O}^*_1(x)\otimes \dots \otimes \mathcal{O}^*_n(x))$,
\item [\textbf{D6})]  Each $D_n$ satisfies the dimension constraint 
\begin{equation}
(\mathbb{N}_d+ \Delta_s) D_n (O_1(x)\otimes \dots \otimes O_n(x)) = \sum_{i=1}^n D_n (O_1(x)\otimes \dots \mathbb{N}_d \mathcal{O}_i(x_i) \otimes \dots  \otimes O_n(x)),
\end{equation}
where $\mathbb{N}_d$ is the dimension counter operator, and $\Delta_s$ is the scaling degree of distributions, 
\item [\textbf{D7})]  Derivatives can be pulled into $D_n$,
\begin{equation}\label{d-D}
d_{x_i} D_n(\mathcal{O}_1(x_1)\otimes \dots \otimes \mathcal{O}_n(x_n)) = D_n(\mathcal{O}_1(x_1) \otimes \dots \otimes d_{x_i} \mathcal{O}_i(x_i) \otimes \dots \otimes \mathcal{O}_n(x_i)).
\end{equation}
\end{enumerate}
Conversely, if $D$ satisfies \textbf{D1} - \textbf{D7}, then any $\tilde{T}$ defined by \eqref{tildeT-T} is a new renormalization scheme.
\end{thm}
\subsubsection*{Conservation of free BRST current}
Similar to the split of the extended action \eqref{hatS0+I}, we also split the BRST current 
\begin{equation}
\textbf{J} = \textbf{J}_0 + \textbf{J}_{\text{int}}.
\end{equation}
We will now show that the quantized $\textbf{J}_0 \in \hat{\bF}_0$, is indeed conserved. This is a prerequisite for formulating the Ward identity \eqref{WI} which in turn will imply the conservation of the full interacting current $\textbf{J}_L$.

The time-ordered products $T_1(\mathcal{O}(x))$ with one factor which satisfy the local and covariance property of $T_n$ are constructed \cite{Hollands:2001nf} as local Wick powers ${:\mathcal{O}(x):}_H$ with respect to a Hadamard parametrix $H^{ij}(x,y)$. A Hadamard parametrix is a distribution defined in a convex normal neighborhood  $U\times U$ of the diagonal in $M \times M$, which is a bi-solution of the free equations of motion modulo $C^\infty(M\times M)$, with a specific wave-front set (see \cite{Radzikowski:1996pa}), and satisfies $\text{Im } H^{ij}(x,y) =\frac{1}{2} \Delta^{ij}(x,y)$. Explicitly, 
\begin{equation}\label{Hadamard}
H^{ij}(x,y) = k_{IJ} \otimes \begin{pmatrix}
  H^{\mu \nu}(x,y) & 0 & 0 & - i \nabla_\nu H^{\mu \nu}(x,y)  \\
  0 &  0 & i H(x,y)  & 0   \\
  0 & - i H(x,y) & 0& 0 \\
 - i \nabla_\mu H^{\mu \nu}(x,y) &0 & 0   & 0\\
 \end{pmatrix},
\end{equation}
where $H(x,y)$, $H^{\mu \nu}(x,y)$ are scalar and vector Hadamard parametrices \cite{DeWitt:2004xz}, given by
\begin{align}
H(x,y) = \frac{1}{2 \pi^2} \left( \frac{u(x,y)}{\sigma+ it0} + v(x,y) \log (\sigma+ it0)\right),
\end{align}
\begin{align}
H_{\mu \nu}(x,y) =   \frac{1}{2 \pi^2} \left(  \frac{u_{\mu \nu}(x,y)}{\sigma+ it0}  + v_{\mu \nu}(x,y) \log(\sigma+ it0) \right).
\end{align}
In the above expressions, $\sigma(x,y)$ is the signed squared geodesic distance between $(x,y) \in U \times U$ and $u, v, u_{\mu \nu}, v_{\mu \nu}$ are smooth functions on $U \times U$ which are determined by requiring $H$ and $H^{\mu \nu}$ to be bi-solutions of the equations of motion up to a smooth reminder.

The important fact about the local Wick powers, is that ${:\mathcal{O}(x):}_H$ differs from ${:\mathcal{O}(x):}_\omega$ only by a smooth function valued in $\hat{\textbf{W}}_0$, see e.g. \cite{Hollands:2007zg} appendix E.
\begin{thm}\label{thm:dJ=0}
In pure Yang-Mills theory,
\begin{enumerate}
\item The local and covariant free BRST current ${: \textbf{J}_0:}_H \in \hat{\bF}_0$ is conserved,

\item the free BRST charge
\beq\label{:Q0:}
{:Q_0:}_H= \int_M \gamma \wedge {: \textbf{J}_0:}_H
\eeq
 (c.f. definition \eqref{BRSTcharge}) squares to zero.
\end{enumerate}
\end{thm}
\begin{proof}
The divergent of the free part of the classical current $\textbf{J}_0$ takes the form
\begin{equation}\label{dJ0=0}
d \textbf{J}_0 = d* d A^I \wedge d C_I - i d B^I * d C_I - i B^I d * d C_I.
\end{equation} 
Since in \eqref{Hadamard} there is no ``contraction'' between either $A^I$ and $C^I$, or $B^I$ and $C^I$, we have ${:d \textbf{J}_0(x):}_H = d\textbf{J}_0(x)$ (the classical current) which vanishes on-shell. For the same reason, we have
\begin{equation}\label{Q02=0}
 {:Q_0:}_H^2 = {:Q_0:}_H \star {:Q_0:}_H=0.
\end{equation}
\end{proof}
\subsection{Auxiliary algebra of renormalized interacting quantum fields}\label{interactingtheory}
 In the previous Section \ref{FreeQFT}, we defined the algebra of free quantum fields $\hat{\textbf{W}}_0(M)$. This algebra is associated to the free classical theory defined by $\hat{S}_0$, the quadratic part of the action functional $\hat{S} = \hat{S}_0+ I$ \eqref{hatS0+I}, Furthermore, the renormalization schemes where defined as distributions valued in $\hat{\textbf{W}}_0(M)$. We now turn to defining the *-algebra of interacting quantum fields $\hat{\textbf{W}}_L(M)$. Here, $L$ is the \textbf{cutoff interaction} 
\beq
\label{interaction-L}
L(f) := \int_M \big(f\lambda \textbf{L}_1 + (f\lambda)^2 \textbf{L}_2 \big),
\eeq
where $f$ is a smooth infra red (IR) cutoff function defined as follows. Let $t: M \rightarrow \mathbb{R}$ be a time function on $M$ whose level surfaces are Cauchy surfaces $\Sigma_t = \{t\} \times \Sigma$ which foliate $M$.  Let $\Sigma = \Sigma_0$. Then, $f$ is equal $1$ on a \emph{time slice}
\beq\label{M_T}
M_T \equiv (-T, T) \times \Sigma,
\eeq
and smoothly falls off to zero outside of $M_T$. Therefore, the true interaction $I$ is obtained by sending $f$ to a constant function (equal $1$) on the whole space-time.

The objects of primary interest in the algebraic approach to interacting QFT are called \emph{interacting fields} $\mathcal{O}(x)_L \in \hat{\textbf{W}}_L$, associated to local-covariant classical functionals $\mathcal{O}(x)$ and the cutoff interaction $L$. The naive limit $f \rightarrow 1$ for interacting fields, however, does not in general exist. In fact, one of the important features of our local framework is that it suffices to construct the algebra of interacting fields for functionals which are localized in a causally convex region\footnote{$\cR \subset M$ is called causally convex if every causal curve with endpoints in $\cR$ entirely lies in $\cR$.} $\cR \subset M_T$.
The \emph{algebraic adiabatic limit} then guarantees that this construction is independent of the chosen cutoff function. We now elaborate on these concepts in the following.   

We begin with defining the interacting analogue of the time-ordered products.
\begin{dfn} Given a renormalization scheme, $T_n$ and a cutoff interaction $L$ of the form \eqref{interaction-L}, 
\begin{enumerate}
\item The \textbf{interacting time-ordered product} of $n$ local functionals $\mathcal{O}_1(x_1),  \dots, \mathcal{O}_n(x_n)$, associated with the cutoff interaction $L$ is the map 
\begin{equation}
T_{L,n} : \textbf{P}(M)^{\otimes n} \rightarrow D'(M^n; \hat{\textbf{W}}_0[[\lambda]])
\end{equation}
defined by the Bogoliubov formula
\begin{align} \nonumber
T_{L,n}(\mathcal{O}_1 \otimes \dots  \otimes \mathcal{O}_n) &:= T(e_\otimes^{iL/\hbar})^{-1} \star  T(\mathcal{O}_1 \otimes \dots \otimes \mathcal{O}_n \otimes e_\otimes^{iL/\hbar} )\\ \label{FO1On}
& = \frac{d^n}{d \tau_1 \dots d\tau_n} T(e_\otimes^{iL/\hbar})^{-1} \star  T(e_\otimes^{iL/\hbar + \tau_1\mathcal{O}_1 + \dots + \tau_n \mathcal{O}_n} )|_{\tau_i=0},
\end{align}
where $T(e_\otimes^{iL/\hbar})= \sum_{n}\frac{i^n}{\hbar^n n!}T_n(L^{\otimes n})$ is the generating functional for the (free) time ordered products, and $T(e_\otimes^{iL/\hbar})^{-1}$ is its formal inverse satisfying
\beq
\label{T^{-1}.T=1}
T(e_\otimes^{iL/\hbar})^{-1} \star T(e_\otimes^{iL/\hbar}) = \textbf{1} = T(e_\otimes^{iL/\hbar}) \star T(e_\otimes^{iL/\hbar})^{-1}.
\eeq

\item For the particular case of one local field, $\mathcal{O}(x)$, the interacting time-ordered product is called an \textbf{interacting field} under interaction $L$ and is denoted
\beq
\label{int-field-def}
\cO_L \equiv T_{L,1}(\mathcal{O}(x)) =  T(e_\otimes^{iL/\hbar})^{-1} \star  T( \mathcal{O}(x) \otimes e_\otimes^{iL/\hbar} ).
\eeq
\end{enumerate}
\end{dfn}
\begin{dfn}
 The off-shell algebra of interacting quantum fields with a cutoff interaction $L$ denoted by $\textbf{W}_L$ is a subalgebra of $\hat{\textbf{W}}_0$ generated by formal power series $T_{L,n}(\mathcal{O}_1 \otimes \dots  \otimes \mathcal{O}_n)$.
\end{dfn}

The interacting time-ordered products can be equivalently written using the \emph{retarded products} \cite{Kallen:1950uha}. 

\begin{dfn} \begin{enumerate}

\item The collection $R= \left(R_{n,k} \right)_{n, k \in \mathbb{N}}$ of multi-linear maps
\begin{equation}
R_{n,k} : \textbf{P}^{\otimes(n+k)} \rightarrow D'(M^{n+k}; \hat{\textbf{W}}_0[[\lambda]]),
\end{equation} 
defined via 
\beq
\label{eq:ret-prod}
 R( \eti{F}; \eti{G}) \defeq T( \eti{G} )^{-1} \star T( \eti{F} \otimes \eti{G}). 
\eeq
is called the \textbf{retarded product}.
\item Denoting the generating functional of interacting time ordered products by
\begin{equation}
\label{gen-func-int-TOP}
T_L(e_\otimes^{iF/\hbar}) := \sum_{n=0} \frac{i^n}{\hbar^n n!} T_{L, n}(F^{\otimes n}),
\end{equation}
the \textbf{interacting retarded products} are defined by
\beq \label{int-ret-prod}
 R_L( \eti{F}; \eti{G}) \defeq T_L( \eti{G} )^{-1} \star T_L( \eti{F} \otimes \eti{G}), 
\eeq
which, in particular, gives
\beq
T_L(\eti{F}) \defeq R(\eti{F}; \eti{L}).
\eeq
\end{enumerate}
\end{dfn}

By causal factorization property of time-ordered products \textbf{T7}, retarded products are trivial if the support of second argument does not intersect the past of the support of the first, i.e.,
\beq \label{suppR}
 R ( \eti{F}; \eti{G}) = T(\eti{F}) \qquad \supp G \cap J^-(\supp F) = \emptyset,
\eeq
where $J^-(\supp F)$ denotes the causal past of support of $F$. It turns out that the following relation holds between interacting time-ordered and retarded products.
\begin{align}
\label{eq:R_1_int}
  R_L(\eti{F}; \cO) &= T_L(\cO \otimes \eti{F}) - \cO_L \star T_L(\eti{F}).
\end{align}

 Furthermore, it follows that \cite{Frob:2018buw} the interacting time-ordered products also satisfy the causal factorization
\beq
\label{eq:CausalFactorization_int}
 T_L(\eti{(F+G)}) = T_L(\eti{F}) \star T_L(\eti{G}) \qquad \supp F \cap J^-(\supp G) = \emptyset.
\eeq

\subsubsection*{Algebraic adiabatic limit}\label{AAL}

So far, we have shown how to (perturbatively) construct the interacting fields $\mathcal{O}(x)_L$ with a local interaction $L$ \eqref{interaction-L}, in a causally convex region $\cR \subset M_T$ where $M_T$ \eqref{M_T} is the region where the cutoff $f =1$. The following important theorem ensures that interacting fields are independent of the cutoff up to unitary equivalence. 

\begin{thm}[\cite{Brunetti:1999jn}]
Let $f$ and $f'$ be two smooth IR cutoffs which coincide on a neighbourhood of $\cR$ and let $L' = L(f')$. Then, there exists a unitary transformation $V_{f, f'}$ such that
\begin{equation}\label{unitaries}
\mathcal{O}_{L'} = V_{f, f'} \star \mathcal{O}_{L} \star V^{-1}_{f, f'},
\end{equation}
and 
\beq\label{cocycle}
V_{f, f'} \star V_{f', f''} \star V_{f'', f} = \textbf{1}. 
\eeq
\end{thm}

Let us denote by $f_T$ a cutoff function which is  compactly supported in $M_{2T}= (-2T, 2T) \times \Sigma $, and is equal to one in $M_T$. Let us also denote by $\cO_{L_T}$ the corresponding interacting field with interaction  $L_T = L(f_T) $, and denote $U_{T} \equiv V_{f_t, f_T}$ for some fixed $t$. Then, from the cocycle condition \eqref{cocycle}, it follows that \cite{Hollands:2002ux} the following sequence is convergent
\begin{align}
\lim_{T \rightarrow \infty} U_T \star \cO(x)_{L_T} \star U_T^{-1}, \quad T \in \mathbb{N}, 
\end{align} 
since it only contains a finite number of terms for each fixed $x$. This defines the \emph{algebraic adiabatic limit}, which intuitively corresponds to fixing the field during the finite time interval $(-t, t) \times \Sigma$, see \cite{Hollands:2002ux} for details. 

 The existence of the algebraic adiabatic limit implies that it is enough to choose cutoff functions which are equal 1, in a neighborhood of $\cR$. Then, the cutoff can be sent to 1 on the entire space-time.   
Put differently, in order to prove statements about $\mathcal{O}_{I}(x)$ with the true interaction $I = L(f=1)$, it suffices to work with the cutoff interaction $L$ where $f=1$ in a sufficiently large neighborhood containing $x$. 

\begin{rmk}\label{remarkT}
\item For $L = \int f \textbf{L}_{\text{int}}$, formally $ \mathcal{S}(L) = T(e_\otimes^{iL/\hbar})$ (the ``local S-matrix'') tends to the S-matrix of the theory in the limit where the cutoff is sent to $1$ on the entire space-time.
Note, however, that $\mathcal{S}$ is not an element of $\hat{\textbf{W}}_L$, and therefore the existence of the ``true S-matrix'' of the theory, cannot be established by the above type of arguments which leads to the existence of the algebraic adiabatic limit  for the interacting fields in $\hat{\textbf{W}}_L$.
 Whether and in which sense such an \textit{adiabatic limit} exists is related to the infra-red properties of the S-matrix whose existence is a non-trivial and difficult task to establish even in Minkowski space-time. Here, we are not concerned with this issue, and we explicitly keep the IR cutoff in all the constructions of the renormalized theory. Therefore in our completely local approach, we disentangle the formulation of the renormalization of quantum fields which is a short distance, and hence UV, issue from the IR issues which does not show up for $\mathcal{S}(L)$.
\end{rmk}

\subsection{Ward identities and anomalies}\label{WardIdentities}
The preservation of the BRST symmetry at quantum level takes the form of a further renormalization condition imposed on $T$; the ``Ward identity'' \cite{Hollands:2007zg}. In order to formulate this identity we need to extend the action of $\hat{s}_0$ on the algebra $\textbf{W}_0$. In fact, using the consistency relations
\begin{equation}
\nabla_\mu \Delta^{\mu \nu}(x,y)= -  \nabla^\nu \Delta(x,y), \hspace{5 mm} \nabla_\nu \Delta^{\mu \nu}(x,y) = - \nabla^\mu \Delta(x,y),
\end{equation}
it follows \cite{Hollands:2007zg} that $\hat{s}_0$ can be consistently extended to $\hat{\textbf{W}}_0$ as a graded derivation, that is, it satisfies the Leibniz rule 
\begin{align}\label{s0-Leibnizrule}
\hat{s}_0 ({\mathcal{O}_1} \star\dots \star {\mathcal{O}_n}) &= \sum_{k} (-1)^{\sum_{l<k} \eps_l}  {\mathcal{O}_1} \star\dots \star \hat{s}_0\mathcal{O}_{k} \star \dots  \star {\mathcal{O}_n}.
\end{align}
and preserves the commutation relations \eqref{Comm-rel-1} in $\hat{\textbf{W}}_0$.

Now, for a given renormalization scheme $T$, and for $F = \int f \wedge \mathcal{O}$ for \emph{all} $\mathcal{O} \in \textbf{P}^p(M)$, $f \in \Omega_0^{4-p}(M)$, the Ward identity takes the form
\footnote{Note that for the renormalization schemes $T_n$ to exist, their arguments have to be local functionals, or cutoff integrated functionals. However, in the expression \eqref{WI}, $\hat{S}_0 = \int \textbf{L}_0$ need not be cut off, since, on account of $(\hat{S}_0, \hat{S}_0)=0$, it appears only in the form $(\hat{S}_0 + F, \hat{S}_0 + F)= 2 (\hat{S}_0, F) + (F,F)$, and $ (\hat{S}_0, F)= \int f(x) \frac{\delta \textbf{L}_0}{\delta \Phi(x)} \frac{\delta \mathcal{O}}{\delta \Phi^\ddag(x)} + \frac{\delta \textbf{L}_0 }{\delta \Phi^\ddag(x)} \frac{\delta \mathcal{O}}{\delta \Phi(x)} $ is cut off with $f$.}
\begin{equation}\label{WI}
[Q_0, T ( e_\otimes^{iF/\hbar}) ] = -\frac{1}{2} T \left( (\hat{S}_0 + F, \hat{S}_0 + F) \otimes e_\otimes^{iF/\hbar} \right), \hspace{5 mm} \text{mod } \mathcal{J}_0.
\end{equation}
In this formula, $Q_0 \equiv T_1(Q_0) = :Q_0:_H$ is the free BRST charge defined above equation \eqref{dJ0=0}, and $(-,-)$ is the anti-bracket bracket defined in \eqref{antibracket}. Note that the graded derivation $[Q_0, -]$ is nilpotent since $[Q_0, [Q_0, -]] = \frac{1}{2}[Q_0^2, -]=0$, by equation \eqref{Q02=0}.

The Ward identity \eqref{WI} is, however, in general violated by a potential \textit{anomaly} term. To define it properly, we express the off-shell violation of \eqref{WI}, where $[Q_0, -]$ is replaced by $i \hbar \hat{s}_0$ which acts non-trivially also on anti-fields, in the following theorem.
\begin{thm}[\textbf{anomalous Ward identity} \cite{Hollands:2007zg}] \label{AnWIthrm} Let $F = \int f \wedge \mathcal{O}$ for all $\mathcal{O} \in \textbf{P}^p(M)$ and $f \in \Omega_0^{4-p}(M)$. Then for a chosen renormalization scheme $T$ it holds
\begin{equation}\label{AnWI}
\hat{s}_0 T ( e_\otimes^{iF/\hbar} )  = \frac{i}{2 \hbar} T \left( (\hat{S}_0 + F, \hat{S}_0 + F) \otimes e_\otimes^{iF/\hbar} \right) + \frac{i}{\hbar} T ( A(e_\otimes^F) \otimes e_\otimes^{iF/\hbar} ).
\end{equation}
The second term in the right hand side defines the anomaly $A(e_\otimes^F)= \sum_n \frac{1}{n!} A_n(F^{\otimes n})$, where each $A_n$ is a map
\begin{equation}
A_n : \textbf{P}(M)^{\otimes n} \rightarrow \textbf{P}^{k_1/ \dots /k_n}(M)[[\hbar]] ,
\end{equation}
with properties
\begin{enumerate}
\item[\textbf{A1})] $A_n (\mathcal{O}_1(x)\otimes \dots \otimes \mathcal{O}_n(x))$ is of order $O(\hbar)$ if all $\mathcal{O}_i$ are of order $O(\hbar^0)$,
\item[\textbf{A2})] Each $A_n$ is locally, and covariantly constructed out of $g$, and is an analytic functional of $g$, 
\item[\textbf{A3})] Each $A_n (\mathcal{O}_1(x)\otimes \dots \otimes \mathcal{O}_n(x))$ is supported on the total diagonal $\Delta_n$,
\item[\textbf{A4})] Each $A_n$ increases the ghost number by one unit,
\item[\textbf{A5})] Each $A_n$ is graded symmetric,
\item[\textbf{A6})] The maps $A_n$ are real $A_n (\mathcal{O}_1(x)\otimes \dots \otimes \mathcal{O}_n(x))^*= A_n (\mathcal{O}^*_1(x)\otimes \dots \otimes \mathcal{O}^*_n(x))$,
\item[\textbf{A7})] Each $A_n$ satisfies the dimension constraint 
\begin{equation}
(\mathbb{N}_d+ \Delta_s) A_n (\mathcal{O}_1(x)\otimes \dots \otimes \mathcal{O}_n(x)) = \sum_{i=1}^n A_n (\mathcal{O}_1(x)\otimes \dots \mathbb{N}_d \mathcal{O}_i(x_i) \otimes \dots  \otimes \mathcal{O}_n(x)),
\end{equation}
\item[\textbf{A8})] Derivatives can be pulled into $A_n$,
\begin{equation}\label{dAnomaly}
d_{x_i} A_n(\mathcal{O}_1(x_1)\otimes \dots \otimes \mathcal{O}_n(x_n)) = A_n(\mathcal{O}_1(x_1) \otimes \dots \otimes d_{x_i} \mathcal{O}_i(x_i) \otimes \dots \otimes \mathcal{O}_n(x_i)),
\end{equation}

\item[\textbf{A9})] \footnote{This property is proven in \cite{Tehrani-Zahn} Appendix B}The anomaly vanishes if one entry contains a basic field or anti-field, i.e.,
\begin{align}
A(\Phi(x) \otimes e_\otimes^F) =0 = A(\Phi^\ddag(x) \otimes e_\otimes^F). 
\end{align}
\end{enumerate}
\end{thm}
Note that the anomalous Ward identity defines the anomaly $A(e_\otimes^F)$ as a map on the space of local actions. This is indeed possible due to the property \eqref{dAnomaly} (see footnote \ref{Fnote2}).


\subsubsection*{Consistency conditions and removal of anomalies}
The study of anomaly turns out to reduce to a cohomological problem. From the anomalous Ward identity \eqref{AnWI} and the nilpotency of $\hat{s}_0$, it follows that \cite{Hollands:2007zg} the anomaly $A(e_\otimes^F)$  satisfies the following \textbf{consistency condition} 
\begin{align} 
\label{cons}
(\hat{S}_0 + F, A(e_\otimes^F)) + A\Big(\{\frac{1}{2} (\hat{S}_0 + F, \hat{S}_0 + F) + A(e_\otimes^F) \} \otimes e_\otimes^F\Big) =0.
\end{align}

In fact, triviality of the cohomology class $H^{4}_1(\hat{s}|d, M )$, which contains potential anomalies, leads to the Ward identity \eqref{WI}. Let us briefly review the argument. 
\begin{thm}[\cite{Hollands:2007zg}]\label{anomaly-removal}
If the cohomology ring $H_1^4(\hat{s} | d)$ is trivial, then there exists a renormalization scheme in which the anomaly $A(e_\otimes^L)$ of the anomalous Ward identity \eqref{AnWI} for the case $F=L$ is absent.
\end{thm}

\begin{proof}[Sketch of proof]
The proof consists of two parts: first, we show that the anomaly vanishes for $f=1$, i.e. $A(e_\otimes^I)=0$ and second, we show that $A(e_\otimes^L)=0$ for all cutoff functions $f$.

For the first part, consider the expansion of $A(e_\otimes^I)$ in powers of $\hbar$
\begin{equation}
A(e_\otimes^I) = A^{m}(e_\otimes^I)\hbar^m + A^{m+1}(e_\otimes^I)\hbar^{m+1} + \dots,
\end{equation}
for some integer $m>0$.
Then, the lowest order in the expansion of equation \eqref{cons} in $\hbar$ implies the ``\textit{$\hbar$-expanded consistency condition}'': 
\begin{equation}\label{hatsA}
\hat{s}  A^{m}(e_\otimes^I)=0. 
\end{equation}
 Now we write $A^{m}(e_\otimes^I) = \int_M a^{m}(x)$ as an integral of a local four-form $a^{m}(x)$. Also, by property $4$ of the definition of anomaly (given in theorem \eqref{AnWIthrm}) $a^m(x)$ has ghost number $1$. Then, \eqref{hatsA} means that $a^{m}(x)$ belongs to the cohomology class $H^{4}_1(\hat{s}|d, M )$, which is trivial. Therefore,
 \begin{equation}\label{a(m)}
a^{m}(x) = \hat{s} b^{m}(x) + d c^{m}(x),
\end{equation}
for some $b^{m}(x) \in \textbf{P}^4_0(M)$ and $c^{m}(x) \in \textbf{P}^3_1(M)$. We can now show that this $\hat{s}$-exact anomaly is absent if we pass to a specific renormalization scheme. Precisely, we need to perform the following steps: 
(1) chose a new scheme $\tilde{T}$ by explicitly constructing the local finite counter terms $D_n$ using $b^{m}$:
\begin{equation}
D_n^{m}(\textbf{L}_1(x_1) \otimes \dots \otimes \textbf{L}_1(x_n)) = - \hbar^m b_n^{m}(x_1) \delta(x_1, \dots ,x_n),
\end{equation}
where $D^{m}$ is the first non-trivial term in the $\hbar$-expansion of $D(e_\otimes^I)$ and where we have expanded $b^{m}= \sum_{n>0} \frac{\lambda^n}{n!}b_n^{m}$,
(2) rewrite he anomalous Ward identity \eqref{AnWI} in the scheme $\tilde{T}$:
\begin{equation}\label{AnWI-tilde}
\hat{s}_0 \tilde{T} ( e_\otimes^{iF/\hbar} )  = \frac{i}{2 \hbar} \tilde{T} \left( (\hat{S}_0 + F, \hat{S}_0 + F) \otimes e_\otimes^{iF/\hbar} \right) + \frac{i}{\hbar} \tilde{T} ( \tilde{A}(e_\otimes^F) \otimes e_\otimes^{iF/\hbar} ).
\end{equation}
which means that in the new scheme, $A(e_\otimes^F)$ is replaced with $\tilde{A}(e_\otimes^F)$,
(3) express the new anomaly $\tilde{A}(e_\otimes^F)$ in terms of the old anomaly $A(e_\otimes^{F+ D(e_\otimes^F)})$ with modified interaction $F+ D(e_\otimes^F)$ ( \cite{Hollands:2007zg}, equation (387)), which to lowest order in $\hbar$ and for $f=1$, takes the form
\begin{equation}\label{tildeA(m)}
\tilde{A}^{m}(e_\otimes^I) = A^{m}(e_\otimes^I) + \hat{s} D^{m}(e_\otimes^I),
\end{equation}
(4) conclude from \eqref{a(m)} and \eqref{tildeA(m)} that
\begin{equation}
\tilde{A}^{m}(e_\otimes^I) = \int_M \tilde{a}^{m}(x) = \int_M a^{m}(x) - \hat{s} b^{m}(x) = \int_M d c^{m}(x) =0.
\end{equation}

For the second part of the proof, we use 
\begin{equation}
0= \tilde{A}^{m}(e_\otimes^I) = \sum_{n} \frac{\lambda^n}{n!} \int \mathcal{A}_n^m(x_1, \dots ,x_n) dx_1 \dots dx_n,
\end{equation}
as a starting point. From this equation it follows that (see \cite{Hollands:2007zg} Lemma 9) $\mathcal{A}_n^m(x_1, \dots ,x_n) = \sum_{k=1}^n d_k \mathcal{C}_{n/k}^m(x_1, \dots ,x_n)$
for some $\mathcal{C}_{n/k}^m \in P^{4/ \dots 3/ \dots 4 }(M^n)$. Using this quantities, we make a further redefinition and pass to another scheme $\hat{T}$ by setting
\begin{equation}\label{hat-D}
\hat{D}_n^{m}(\textbf{L}_1(x_1) \otimes \dots \otimes \textbf{K}_1(x_k) \otimes \dots \otimes \textbf{L}_1(x_n)) = - \hbar^m \mathcal{C}_{n/k}^m (x_1) \delta(x_1, \dots ,x_n),
\end{equation}
where $\textbf{K}_1 \in \textbf{P}_1^3(M)$ is defined by $\hat{s}_0 \textbf{L}_1 = d \textbf{K}_1$. Again, we express the lowest order in the expansion of anomaly in the scheme $\hat{T}$ in terms of that in the scheme $\tilde{T}$:
\begin{equation}
\hat{A}^m(e_\otimes^L) = \tilde{A}^m(e_\otimes^L) + \frac{1}{2}\hat{D}((\hat{S}_0 + L, \hat{S}_0 + L) \otimes e_\otimes^L).
\end{equation}
However, using \eqref{hat-D} we have $\hat{D}^m((L,L) \otimes e_\otimes^L)=0$ and $\hat{D}^m(\hat{s}_0 {L} \otimes e_\otimes^L)= - \tilde{A}^m(e_\otimes^L)$.
This means that there exists a scheme $\hat{T}$ in which the anomaly $\hat{A}(e_\otimes^L)$ vanishes to lowest order in $\hbar$ and to all orders in $\lambda$. Proceeding by iteration in higher powers of $\hbar$, we conclude that the anomaly can be removed to all orders in $\hbar$ and $\lambda$.
\end{proof}
\begin{rmk}\label{YM-anomaly}
For the pure Yang-Mills case, $H^{4}_1(\hat{s}|d, M )$ is not trivial. In fact, when $G$ is semi-simple with no abelian factors, this cohomology class is generated by the so-called ``gauge anomaly''  \cite{Barnich:2000zw} of the form
\begin{equation}\label{gaugeanomaly}
\mathcal{A} = d C^I \wedge \left( d_{IJK} A^J \wedge d A^K - \frac{1}{12} D_{IJKL} A^J \wedge A^K \wedge A^L \right),
\end{equation} 
where $d_{IJK} \in \g^{\otimes 3}$ and $d_{IJKL} \in \g^{\otimes 4}$ are $\g$-invariant totally symmetric tensors in some representations of the Lie algebra $\mathfrak{g}$.
Nevertheless, for pure Yang-Mills theory with a generic gauge group one can still argue \cite{Hollands:2007zg} that $a^{m}(x)$ is the zero element in this cohomology class as follows. In Minkowski space-time, where parity is an isometry, one can argue that $a^m(x)$ is parity odd, i.e. it transforms as $a^m \mapsto -a^m$ under $d x \mapsto - d x$. However, the gauge anomaly \eqref{gaugeanomaly} is evidently even under parity. Therefore $a^m(x)$ is the zero element in the cohomology. On the other hand, according to the equation (47) of \cite{Hollands:2007zg}, at dimension $4$ and ghost number $1$ there cannot be any space-time curvature contribution to the gauge anomaly. Thus, since anomaly is a local-covariant quantity, when it vanishes in one space-time it vanishes on all space-times. 
\end{rmk}

\section{Quantum BRST charge and the algebra of interacting quantum fields}

\subsection{Conservation of the renormalized BRST Noether current}\label{dJ=0}

We have already shown in section \ref{FreeQFT} that the free part $\textbf{J}_0$ of the BRST current $\textbf{J}$ is conserved in the free theory.
In this section, we review the argument in \cite{Hollands:2007zg} which leads to the conservation of the full interacting BRST current. We show that once a renormalization scheme is chosen in which the anomaly $A(e_\otimes^L)$ vanishes, the interacting field corresponding to Noether current of BRST symmetry $\textbf{J}(x)$ is conserved as a consequent of the Ward identity \eqref{WI}. The important point about the proof is that it holds true irrespective of the functional form of the classical BRST current $\textbf{J}(x)$, and therefore is satisfied for a wider class of theories with local gauge symmetry which admit an appropriate BRST formulation.  We first state the following lemma which is needed in the proof of our main Theorem \ref{[QL,-]}.

\begin{lemma}[\cite{Hollands:2007zg}]\label{lemma-T(dJ0+int)}
 Let $L$ be the cutoff interaction \eqref{interaction-L}. If the renormalization scheme satisfies $A(e_\otimes^L)=0$, then there exist another renormalization scheme in which the following identity holds for all $x \in M$
\begin{align}  \label{dJ=0-proof}
 T \big( \{ d \textbf{J}_0(x) + \hat{\hs}_0 L(x) + (L, L)(x)  \}   \otimes e_\otimes^{i L/\hbar} \big) \eqos 0, 
\end{align}
where
\begin{align}
\hat{\hs}_0 L(x)& =  (\hat{S}_0, \Phi^i(x)) (\Phi_i^\ddag(x), L) + (\Phi \leftrightarrow \Phi^\ddag) \\ 
(L, L)(x) & =  ( L, \Phi^i(x)) (\Phi_i^\ddag(x), L).
\end{align}
\end{lemma}

\begin{thm} [\cite{Hollands:2007zg}]\label{dJ-int=0}
In a renormalization scheme such that $A(e_\otimes^L)=0$ and the identity \eqref{dJ=0-proof} holds, the interacting BRST current is conserved on-shell.
\end{thm}
\begin{proof}
According to the discussion of the algebraic adiabatic limit in Section \ref{interactingtheory}, it suffices to prove
\begin{equation}\label{T(dJ)=0}
 T \big( d \textbf{J}(x)   \otimes e_\otimes^{i L/\hbar} \big)=0, \hspace{5 mm}  \forall x \in M_T,
\end{equation}
 where $M_T = (-T, T) \times \Sigma$ is the region where $f=1$.

We first look at the divergence of the classical BRST current which using \eqref{dJ(x)} and the fact that the cutoff function is equal $1$ in $M_T$ can be written as
\begin{equation}\label{dJinM_T}
d \textbf{J} (x) =   (\hat{S}_0 + L, \Phi^i(x)) (\Phi_i^\ddag(x), \hat{S}_0 + L), \hspace{5 mm} \forall x \in M_T,
\end{equation}
Expanding in powers of $\lambda$, this leads to
\begin{align}
\label{dJ_0inM_T}
d \textbf{J}_0(x) &=  (\hat{S}_0, \Phi^i(x)) (\Phi_i^\ddag(x), \hat{S}_0) \hspace{5 mm} \forall x \in M_T, \\ 
\label{dJ_>0inM_T}
d \textbf{J}_{\text{int}}(x) &=  (\hat{S}_0, \Phi^i(x)) (\Phi_i^\ddag(x), L) + (\Phi \leftrightarrow \Phi^\ddag) + ( L, \Phi^i(x)) (\Phi_i^\ddag(x), L),\hspace{5 mm} \forall x \in M_T.
\end{align}
Thus, using the identity \eqref{dJ=0-proof} proved in Lemma \ref{lemma-T(dJ0+int)} we find that for all $ x \in M_T$ we have
\begin{align}\nonumber
 T ( d \textbf{J}(x)   \otimes e_\otimes^{i L/\hbar} ) & =  T ( \{ d \textbf{J}_0(x) + d \textbf{J}_\ia (x) \}   \otimes e_\otimes^{i L/\hbar} ) \\ \nonumber
 & = T \big( \{ d \textbf{J}_0(x) + \hat{\hs}_0 L(x) + (L, L)(x)  \}   \otimes e_\otimes^{i L/\hbar} \big) \\
 & \eqos 0.
\end{align}

\end{proof}

\subsection{Action of $[\textbf{Q}_L, -]$ on quantum fields}
In this section, we derive our main result concerning the commutator of the quantum BRST charge and interacting time-ordered products in Theorem \ref{[QL, OL]}.
Note that since the current $\textbf{J}_L$ is conserved on-shell, the corresponding BRST charge $\textbf{Q}_L = \int \gamma \wedge \textbf{J}_L$, where $\gamma$ is defined in below equation \eqref{BRSTcharge}, is independent of the precise choice of $\gamma$ up to an element in $\mathcal{J}_0$, i.e., when considered as an equivalence class in the \textbf{on-shell interacting algebra}
\beq \label{int-on-shell-alg}
\hat{\bF}_L := \hat{\bW}_L / \mathcal{J}_0.
\eeq

In the following theorem, we derive an expression for $[\textbf{Q}_L, -]: \hat{\bF}_L \rightarrow \hat{\bF}_L$ when acting on $T_{L,n}(\mathcal{O}_1 \otimes \dots \otimes  \mathcal{O}_n)$ in terms of $[{Q}_0, T_{L,n}(\mathcal{O}_1 \otimes \dots \otimes  \mathcal{O}_n)]$, with $Q_0$ being the free BRST charge, and in Theorem \ref{[QL, OL]} we express this commutator in terms of the quantum BRST operator \eqref{hat-q}. Next, in Section \ref{sectionQ2=0}, we prove that $[\textbf{Q}_L, -]$ is nilpotent.

\begin{thm}\label{[QL,-]} 
In a renormalization scheme such that $A(e_\otimes^L)=0$, it holds
\begin{align}\nonumber
[ \textbf{Q}_L, T_{L,n}(\mathcal{O}_1 \otimes \dots \otimes  \mathcal{O}_n)] &\eqos [Q_0, T_{L,n}(\mathcal{O}_1 \otimes \dots \otimes \mathcal{O}_n)] \\ \label{[QL,O1 ... On]}
& \quad +  R_{L}(\mathcal{O}_1 \otimes \dots \otimes \mathcal{O}_n; \frac{1}{2} (\hat{S}_0 +L,\hat{S}_0 +L)),
\end{align}
where $Q_0$ is the BRST charge of the free theory.
\end{thm}
\begin{proof}
According to our discussion of algebraic adiabatic limit, the functionals $\mathcal{O}_1 , \dots,  \mathcal{O}_n$ are supported in a causally convex region $\cR \subset M_T$, where $M_T$ \eqref{M_T} is the region where the cutoff function $f=1$. Since $\textbf{J}_L(x)$ is conserved for $x \in M_T$ (Theorem \ref{dJ-int=0}), in the definition of the charge $\textbf{Q}_L = \int \gamma \wedge \textbf{J}_L$, the one-form $\gamma$ must be supported in $M_T$. Keeping in mind that $\textbf{Q}_L$ is independent of $\gamma$, we choose two such one-forms $\gamma_+$ and $\gamma_-$ which are supported in the future and past of $ \cR$. We denote their difference by $ d h = \gamma_+ - \gamma_-$ for some smooth compactly supported function $h$ which is equal $1$ on $\cR$. This is depicted in Figure \ref{fig:Localization}.
  
Then, using the causal factorization of interacting time-ordered products \eqref{eq:CausalFactorization_int} we get
 \begin{align} \nonumber 
 [\textbf{Q}_L,  T_{L,n}(\mathcal{O}_1 \otimes \dots \otimes  \mathcal{O}_n)] & = \textbf{Q}_L \star T_{L,n}(\mathcal{O}_1 \otimes \dots \otimes  \mathcal{O}_n)\\ \nonumber
 & \quad  -  T_{L,n}(\mathcal{O}_1 \otimes \dots \otimes  \mathcal{O}_n) \star \textbf{Q}_L\\ \nonumber 
&=   T_{L,1}( \int \gamma_-(x) \wedge \textbf{J}(x)) \star  T_{L,n}(\mathcal{O}_1 \otimes \dots \otimes  \mathcal{O}_n)  \\ \nonumber 
& \quad - T_{L,n}(\mathcal{O}_1 \otimes \dots \otimes  \mathcal{O}_n)  \star T_{L,1}( \int \gamma_+(x) \wedge \textbf{J}(x))\\ \nonumber
&= T_{L,n+1}( \int (\gamma_-(x) - \gamma_+(x)) \wedge \textbf{J}(x) \otimes \mathcal{O}_1 \otimes \dots \otimes  \mathcal{O}_n) \\ \nonumber
&= T_{L, n+1}( \int h(x) d \textbf{J}(x) \otimes  \mathcal{O}_1 \otimes \dots \otimes  \mathcal{O}_n)\\ \nonumber
&=T_{L, n+1}( \int h(x) d \textbf{J}_0(x) \otimes  \mathcal{O}_1 \otimes \dots \otimes  \mathcal{O}_n)\\ 
&  + T_{L, n+1}( \int h(x)( \hat{\hs}_0 L(x) + (L, L)(x)) \otimes  \mathcal{O}_1 \otimes \dots \otimes  \mathcal{O}_n),
\end{align}
where we have used the causal factorization of interacting time-ordered products \eqref{eq:CausalFactorization_int} and we have used the identity \eqref{dJ_>0inM_T} which states $d \textbf{J}(x) = d \textbf{J}_0(x) + \hat{\hs}_0 L(x) + (L, L)(x)$ for $x \in M_T$. 

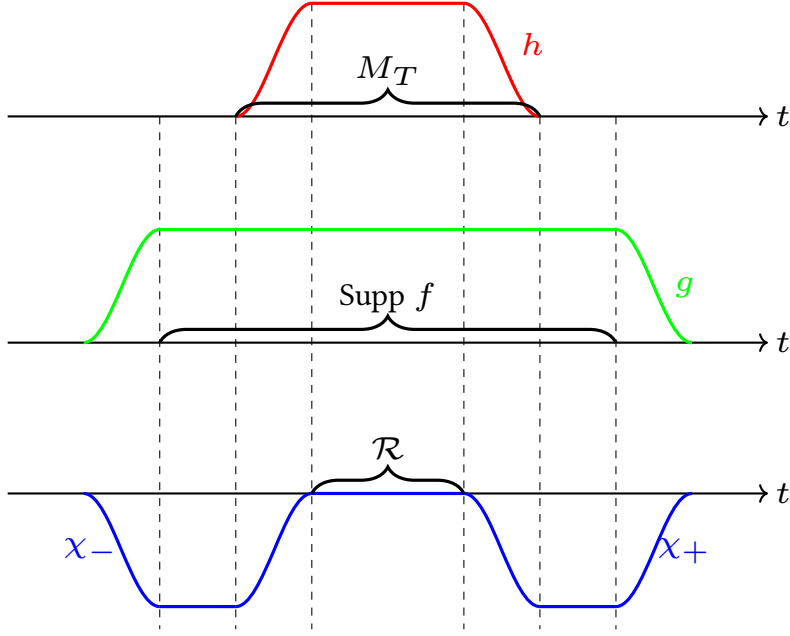
\begin{figure}
\begin{centering}
\begin{tikzpicture}[very thick, scale=2, every node/.style={transform shape}]

\draw[->, thick, color=black ] (3,0) -- (8,0) ;
\filldraw (8.1, 0)  node {\tiny{$t$}}; 
\draw[->, thick, color=black] (3,-1.5) -- (8,-1.5) ;
\filldraw (8.1, -1.5)  node {\tiny{$t$}}; 
\draw[->, thick, color=black] (3,-2.5) -- (8,-2.5) ;
\filldraw (8.1, -2.5)  node {\tiny{$t$}}; 

\draw [thin, black, dashed] (5, 0.75) -- (5, -3.40) ;
\draw [thin, black, dashed] (6, 0.75) -- (6, -3.40) ;
\draw [thin, black, dashed] (4.5, 0) -- (4.5, -3.40) ;
\draw [thin, black, dashed] (6.5, 0) -- (6.5, -3.40) ;
\draw [thin, black, dashed] (4, 0) -- (4, -3.40) ;
\draw [thin, black, dashed] (7, 0) -- (7, -3.40) ;

\filldraw[red] (6.45, 0.48)  node {\tiny{$h$}}; 
\draw[very thick, color=red ] (5, 0.75) -- (6, 0.75) ;
\draw[very thick, color=red ] (4.5,0) cos (4.75, 0.38) sin (5, 0.75)  ;
\draw[very thick, color=red ] (6, 0.75) cos (6.25, 0.38) sin (6.5, 0)  ;

\filldraw[green] (7.45, -1.13 )  node {\tiny{$g$}}; 
\draw[very thick, color=green] (4, -0.75) -- (7, -0.75) ;
\draw[very thick, color=green] (3.5, -1.5 ) cos (3.75, -1.13 ) sin (4,-0.75)  ;
\draw[very thick, color=green] (7, -0.75) cos (7.25, -1.13 ) sin (7.5, -1.5)  ;

\draw[very thick, color=blue] (5,-2.5) -- (6,-2.5) ;
\draw[very thick, color=blue ] (4.5, -3.25) cos (4.75, -2.88) sin (5, -2.5)  ;
\draw[very thick, color=blue ] (6, -2.5) cos (6.25, -2.88) sin (6.5, -3.25)  ;

\filldraw[blue]  (3.55, -2.88)  node {\tiny{$\chi_-$}}; 
\draw[very thick, color=blue] (4,-3.25) -- (4.5,-3.25) ;
\draw[very thick, color=blue ] (3.5, -2.5) cos (3.75, -2.88) sin (4, -3.25)  ;

\filldraw[blue]  (7.45, -2.88)  node {\tiny{$\chi_+$}};
\draw[very thick, color=blue] (6.5,-3.25) -- (7,-3.25) ;
\draw[very thick, color=blue ] (7, -3.25) cos (7.25, -2.88) sin (7.5, -2.5)  ;

\draw [decorate, decoration={brace,amplitude=10pt},xshift=0pt,yshift=0pt] (4.5,0) -- (6.5,0) node [black,midway,yshift=0.3 cm] {\tiny{$M_T$}};
\draw [decorate, decoration={brace,amplitude=10pt},xshift=0pt,yshift=0pt] (5,-2.5) -- (6,-2.5) node [black,midway,yshift=0.3 cm] {\tiny{$\cR$}}; 
\draw [decorate, decoration={brace,amplitude=10pt},xshift=0pt,yshift=0pt] (4,-1.5) -- (7,-1.5)  node [black,midway,yshift=0.3 cm] {\tiny{Supp $f$}}; 

\end{tikzpicture}
\caption{\label{fig:Localization} Different cutoff functions used in the proof of Theorem \eqref{[QL,-]}. }
\end{centering}
\end{figure}

We now want to write the term $T_{L, n+1}( \int h(x) d \textbf{J}_0(x) \otimes  \mathcal{O}_1 \otimes \dots \otimes  \mathcal{O}_n)$ in terms of the commutator with the free BRST charge $Q_0$. Since by Theorem \ref{thm:dJ=0} the free current $\textbf{J}_0(x)$ is conserved for all $x \in M$, in the definition of the free charge $Q_0= \int \eta \wedge \textbf{J}_0$ the one-form $\eta$ can be supported everywhere. We, then, choose another pair of one-forms $\eta_\pm$ supported in the past and future of the support of $L$, and let $d g = \eta_+ - \eta_-$ with $g$ a smooth function which is equal $1$ on support of $f$ (see Figure \ref{fig:Localization}). Thus using the causal factorization of time-ordered products \eqref{CFaxiom}, we get
\begin{align}\nonumber
T_{L, n+1}( \int g d \textbf{J}_0 \otimes \mathcal{O}_1 \otimes \dots \otimes  \mathcal{O}_n ) & =  T(e_\otimes^{iL/\hbar})^{-1} \star  T( \int (\eta_- -\eta_+) \wedge \textbf{J}_0 \otimes \mathcal{O}_1 \otimes \dots \otimes  \mathcal{O}_n \otimes e_\otimes^{iL/\hbar})\\
&= T(e_\otimes^{iL/\hbar})^{-1} \star  [Q_0, T(\mathcal{O}_1 \otimes \dots \otimes  \mathcal{O}_n \otimes e_\otimes^{iL/\hbar})].
\end{align}
Also since $g=1$ on Supp $f$, we get
\begin{align}\nonumber
&T_{L, n+1}( \int g(x)( \hat{\hs}_0 L(x) + (L, L)(x)) \otimes  \mathcal{O}_1 \otimes \dots \otimes  \mathcal{O}_n) \\ 
&= T_{L, n+1}( \frac{1}{2} (\hat{S}_0 +L,\hat{S}_0 +L) \otimes  \mathcal{O}_1 \otimes \dots \otimes  \mathcal{O}_n).
\end{align}

Thus, we obtain
 \begin{align} \nonumber 
 [\textbf{Q}_L,  T_{L,n}(\mathcal{O}_1 \otimes \dots \otimes  \mathcal{O}_n)] & =  T(e_\otimes^{iL/\hbar})^{-1} \star  [Q_0, T(\mathcal{O}_1 \otimes \dots \otimes  \mathcal{O}_n \otimes e_\otimes^{iL/\hbar})] \\ \nonumber
 &  + T_{L, n+1}( \frac{1}{2} (\hat{S}_0 +L,\hat{S}_0 +L) \otimes  \mathcal{O}_1 \otimes \dots \otimes  \mathcal{O}_n)\\ \nonumber
& + T_{L, n+1}( \int (h -g)(x) d \textbf{J}_0(x) \otimes  \mathcal{O}_1 \otimes \dots \otimes  \mathcal{O}_n)\\ \label{[QL,-]-proof}
&  + T_{L, n+1}( \int (h-g)(x)( \hat{\hs}_0 L(x) + (L, L)(x)) \otimes  \mathcal{O}_1 \otimes \dots \otimes  \mathcal{O}_n).
\end{align}
The sum of the first two terms on the right hand side of the above equation becomes 
\begin{align}\nonumber
& T(e_\otimes^{iL/\hbar})^{-1} \star  [Q_0, T(\mathcal{O}_1 \otimes \dots \otimes  \mathcal{O}_n \otimes e_\otimes^{iL/\hbar})] + T_{L, n+1}( \frac{1}{2} (\hat{S}_0 +L,\hat{S}_0 +L) \otimes  \mathcal{O}_1 \otimes \dots \otimes  \mathcal{O}_n)\\ \nonumber
& = [Q_0, T_{L,n}(\mathcal{O}_1 \otimes \dots \otimes  \mathcal{O}_n )] -  [Q_0, T(e_\otimes^{iL/\hbar})^{-1}] \star T(\mathcal{O}_1 \otimes \dots \otimes  \mathcal{O}_n \otimes e_\otimes^{iL/\hbar}) \\ \nonumber
& \quad + T_{L, n+1}( \frac{1}{2} (\hat{S}_0 +L,\hat{S}_0 +L) \otimes  \mathcal{O}_1 \otimes \dots \otimes  \mathcal{O}_n),
\end{align}
which upon using \eqref{WI} for $F=L$ and the relation \eqref{eq:R_1_int}, becomes the r.h.s. of \eqref{[QL,O1 ... On]}. It therefore remains to show that the sum of the third and fourth terms in \eqref{[QL,-]-proof} vanishes. To prove this, we note that $h-g$ can be written as $h-g = \chi_- + \chi_+$ where $\chi_\pm$ are smooth functions supported in the past and future of $\cR$.  Since  $\mathcal{O}_1 , \dots,  \mathcal{O}_n$ are supported in $\cR$, we can use the causal factorization of interacting time-ordered products and write
 \begin{align} \nonumber 
 &T_{L, n+1}( \int (h-g)(x)( d \textbf{J}_0(x) + \hat{\hs}_0 L(x) + (L, L)(x)) \otimes  \mathcal{O}_1 \otimes \dots \otimes  \mathcal{O}_n) \\ \nonumber
 & =  T_{L,1}( \int \chi_-(x)( d \textbf{J}_0(x) + \hat{\hs}_0 L(x) + (L, L)(x)) \star T_{L, n+1}( \mathcal{O}_1 \otimes \dots \otimes  \mathcal{O}_n)\\ \nonumber
 & + T_{L, n+1}( \mathcal{O}_1 \otimes \dots \otimes  \mathcal{O}_n) \star T_{L,1}( \int \chi_+(x)( d \textbf{J}_0(x) + \hat{\hs}_0 L(x) + (L, L)(x)) \\\label{[QL,-]-proof2}
 & =0,
\end{align}
where we have used \eqref{dJ=0-proof}.
\end{proof}

Theorem \ref{[QL,-]} gives an expression for the commutator with the interacting BRST charge in terms of the free BRST charge and an extra term. We would now like to derive an explicit formula for the commutator of $\textbf{Q}_L$ and quantum fields without any reference to $Q_0$, and express it in terms of the classical BRST differential $\hat{\hs}$ plus quantum corrections of order $O(\hbar)$ coming from the anomalies with higher insertions. This is given in the following theorem.

\begin{thm}\label{[QL, OL]}
In a renormalization scheme such that $A(e_\otimes^L)=0$, the following \textbf{interacting anomalous Ward identity} holds for $F = \int f \wedge \mathcal{O}$ for \emph{all} $\mathcal{O} \in \textbf{P}^p(M)$, $f \in \Omega_0^{4-p}(M)$,
\begin{align}\label{int-AWI}
[\textbf{Q}_L, T_L(e_\otimes^{iF/\hbar})] \eqos -  T_L \left( \{ \frac{1}{2} (\hat{S} + F, \hat{S} + F) + \hat{A}(e_\otimes^F) \}\otimes e_\otimes^{iF/\hbar} \right),
\end{align}
where 
\beq\label{hat-A-gen-func}
\hat{A}(e_\otimes^F) = \sum_{n=1} \frac{1}{n!} \hat{A}_n(F^{\otimes n}),
\eeq 
is the generating functional of \textbf{interacting anomalies} $\hat{A}_n$ defined by
\beq \label{hat-A-n}
\hat{A}_n( \mathcal{O}_1 \otimes \dots \otimes \mathcal{O}_n) =  A( \mathcal{O}_1 \otimes \dots \otimes \mathcal{O}_n \otimes e_\otimes^{L}).
\eeq 

\end{thm}

\begin{proof}  
From Theorem \ref{[QL,-]}, we obtain
\begin{align}\label{[QL, TL]}
[\textbf{Q}_L, T_L(e_\otimes^{iF/\hbar})] \eqos [{Q}_0, T_L(e_\otimes^{iF/\hbar})] + R_L(e_\otimes^{iF/\hbar}; \frac{1}{2}(\hat{S}_0 + L, \hat{S}_0 + L)). 
\end{align}
It thus remains to show that 
\begin{align}\nonumber
[{Q}_0, T_L(e_\otimes^{iF/\hbar})] & \eqos -  T_L \big( \{ \frac{1}{2} (\hat{S} + F, \hat{S} + F) + \hat{A}(e_\otimes^F) \}\otimes e_\otimes^{iF/\hbar} \big) \\ \label{[Q0, TL]}
& \quad - R_L(e_\otimes^{iF/\hbar}; \frac{1}{2}(\hat{S}_0 + L, \hat{S}_0 + L)),
\end{align}
which upon expanding the generating functional $T_L(e_\otimes^{i F/ \hbar})$ is equivalent to
\begin{align} 
\nonumber
  \frac{1}{i \hbar}[{Q}_0,T_{L,n}(\mathcal{O}_1 \otimes \dots  \otimes \mathcal{O}_n)] = &  \sum_{k=0} T_{L,n}(  \mathcal{O}_1  \otimes \dots \otimes \hat{\hs} \mathcal{O}_i \otimes \dots \otimes \mathcal{O}_n)\\ \nonumber
& + \frac{\hbar}{i} \sum_{I_2}  T_{L,n-1}(  {(\mathcal{O}_i, \mathcal{O}_j)}_{i, j \in I_2} \otimes \bigotimes_{k \in I_2^c}  \mathcal{O}_k)\\ \nonumber
&+ \sum_{I} (\frac{\hbar}{i})^{|I|-1}  \Big( \hat{A}_{|I|}(\bigotimes_{i \in I} \mathcal{O}_i) \otimes \bigotimes_{j \in I^c}  \mathcal{O}_j \Big)\\ \label{[Q0, T(O1-On)]}
& - R_{L}(\mathcal{O}_1 \otimes \dots \otimes \mathcal{O}_n; \frac{1}{2} (\hat{S}_0 +L,\hat{S}_0 +L)),
\end{align}
where $I$ is a non-empty and ordered partition of the set $\{ 1, 2, \dots, n \}$ and $I^c$ is the complement partition and $|I_2|=2$. 
We will argue that the proof of the above identity follows from the anomalous Ward identity \eqref{AnWI}: For $n=1$, we calculate using \eqref{AnWI}
\begin{align}\nonumber
\hat{\hs}_0 T( \mathcal{O} \otimes e_\otimes^{iF/ \hbar}) &= \frac{\hbar}{i} \frac{d}{d \tau} \hat{\hs}_0 T(e_\otimes^{i (F + \tau \mathcal{O})/ \hbar}) \Big|_{\tau=0} \\ \nonumber
&= \frac{\hbar}{i} \frac{d}{d \tau} \left( \frac{i}{2\hbar}T( (\hat{S}_0 + F + \tau \mathcal{O}, \hat{S}_0 + F + \tau \mathcal{O}) \otimes e_\otimes^{i (F + \tau \mathcal{O})/ \hbar}) \right) \Big|_{\tau=0}\\ \nonumber
& \quad + \frac{\hbar}{i} \frac{d}{d \tau} \left( \frac{i}{\hbar}T( A(e_\otimes^{F + \tau \mathcal{O}}) \otimes e_\otimes^{i (F + \tau \mathcal{O})/ \hbar}) \right) \Big|_{\tau=0}\\ \nonumber
& = T(\big( (\hat{S}_0 + F, \mathcal{O}) + A(e_\otimes^F \otimes \mathcal{O}) \big) \otimes e_\otimes^{i F/ \hbar}) \\ \label{s0-T(O)}
& \quad +  \frac{i}{\hbar} T( \mathcal{O} \otimes \big( \frac{1}{2}(\hat{S}_0 + F, \hat{S}_0 + F) + A(e_\otimes^F) \big) \otimes  e_\otimes^{i F/ \hbar}).
\end{align}
We, therefore, have
\begin{align}\nonumber
\hat{\hs}_0 \mathcal{O}_L &=  \hat{\hs}_0 T(e_\otimes^{iL/ \hbar})^{-1} \star  T(\mathcal{O} \otimes e_\otimes^{iL/ \hbar}) + T(e_\otimes^{iL/ \hbar})^{-1} \star \hat{\hs}_0 T( \mathcal{O} \otimes e_\otimes^{iL/ \hbar}) \\ \nonumber
&= - T(e_\otimes^{iL/ \hbar})^{-1} \star \hat{\hs}_0 T(e_\otimes^{iL/ \hbar}) \star \mathcal{O}_L + T(e_\otimes^{iL/ \hbar})^{-1} \star \hat{\hs}_0 T(\mathcal{O} \otimes e_\otimes^{iL/ \hbar} ) \\ \label{s0PsiF}
&= \left( (\hat{S}_0 +L, \mathcal{O}) + A(e_\otimes^L \otimes \mathcal{O} )\right)_L +  \frac{i}{\hbar}  R (\mathcal{O}; \frac{1}{2} (\hat{S}_0 +L,\hat{S}_0+L) \otimes e_\otimes^{iL/\hbar}),
\end{align}
where we have again used the anomalous Ward identity \eqref{AnWI}, and equation \eqref{s0-T(O)} for the case $F=L$, and $A(e_\otimes^L)=0$.
 Going on-shell, that is setting $\hat{\hs}_0 \eqos \frac{1}{i \hbar}[Q_0, -]$, we obtain 
 \begin{align}
[Q_0, \mathcal{O}_L] \eqos i \hbar \left( (\hat{S}, \mathcal{O}) + \hat{A}_1(\mathcal{O} )\right)_L - R (\mathcal{O}; \frac{1}{2} (\hat{S}_0 +L,\hat{S}_0+L) \otimes e_\otimes^{iL/\hbar}),
\end{align}
where in the first term on the right hand side, we have used that $f=1$ on $\cR$ where $\cO$ is localized. This is \eqref{[Q0, T(O1-On)]} for $n=1$. 
For $n>1$, we replace $F$ with $F+ \tau_1 \mathcal{O}_1 + \dots \tau_n \mathcal{O}_n$ in \eqref{AnWI}, differentiate with respect to $\tau_1 \dots \tau_n$ and set $\tau_i =0$. This procedure, together with the following relations
\begin{align}\nonumber 
&\frac{d}{d \tau}\big( \frac{1}{2} (\hat{S}_0+F+ \tau \mathcal{O},\hat{S}_0 +F+ \tau \mathcal{O}) +A(e_\otimes^{F+ \tau \mathcal{O}})  \big) \big|_{\tau=0} =  (\hat{S}_0 + F, \mathcal{O})+ A( \mathcal{O} \otimes e_\otimes^F),\\ \label{tau-O-(1)}
& \frac{d}{d \tau}\big((\hat{S}_0 + F + \tau \mathcal{O}_2, \mathcal{O}_1)+ A( \mathcal{O}_1 \otimes e_\otimes^{F + \tau \mathcal{O}_2})  \big) \big|_{\tau=0} = (\mathcal{O}_1, \mathcal{O}_2)+ A( \mathcal{O}_1 \otimes \mathcal{O}_2 \otimes e_\otimes^F),\\ \label{tau-O-(2)}
& \frac{d}{d \tau} A( \mathcal{O}_1 \otimes \dots \otimes \mathcal{O}_k \otimes e_\otimes^{F+ \tau \mathcal{O}_{k+1}}) \big|_{\tau=0} = A( \mathcal{O}_1 \otimes \dots \otimes \mathcal{O}_{k+1} \otimes e_\otimes^{F}),\\ \label{tau-O-(3)}
\end{align}
leads, after straightforward calculations, to
\begin{align}  \nonumber
&\hat{\hs}_0 T(\mathcal{O}_1 \otimes \dots \otimes  \mathcal{O}_n \otimes e_\otimes^{iF/\hbar}) \\ \nonumber
=& \sum_{i=1}^n T (  \mathcal{O}_1  \otimes \dots \otimes (\hat{S}_0 + F, \mathcal{O}_i)  \otimes \dots \otimes \mathcal{O}_n \otimes e_\otimes^{iF/\hbar})\\  \nonumber
& + \frac{\hbar}{i} \sum_{I_2}^{n} T_{L,n-1}(  {(\mathcal{O}_i, \mathcal{O}_j)}_{i, j \in I_2} \otimes \bigotimes_{k \in I_2^c}  \mathcal{O}_k \otimes e_\otimes^{iF /\hbar}) \\  \nonumber
& +\sum_{I} (\frac{\hbar}{i})^{|I|-1}  T_{n - |I|+1} \Big(  A_{|I|}(\bigotimes_{i \in I} \mathcal{O}_i \otimes e_\otimes^{iF/\hbar}) \otimes \bigotimes_{j \in I^c}  \mathcal{O}_j \otimes e_\otimes^{iF/\hbar} \Big),\\ \label{Q0TF}
&  - T ( \big( \frac{1}{2}(\hat{S}_0 +F,\hat{S}_0 +F) + A(e_\otimes^F) \big)\otimes \mathcal{O}_1 \otimes \dots \otimes \mathcal{O}_n \otimes  e_\otimes^{iF/\hbar})=0.
\end{align}
Now setting $F=L$, $A(e_\otimes^L)=0$ and $\hat{\hs}_0 \eqos \frac{1}{i \hbar}[Q_0, -]$, we arrive at \eqref{[Q0, T(O1-On)]}.
\end{proof}

\begin{rmk}\label{remark-Q}
\begin{enumerate}

\item Equation \eqref{[Q0, TL]} expresses that the commutator of the free BRST charge $Q_0$ and interacting time-ordered products of local functional $F$ (localized in $\cR$) is given by two terms. The first one contains $\frac{1}{2} (\hat{S} + F, \hat{S} + F) + \hat{A}(e_\otimes^F)$  which are localized in $\cR$ where the cutoff $f=1$. However, in the second term, one is not  allowed to set $f=1$ since $(\hat{S}_0 +L,\hat{S}_0+L)$ is not localized in $\cR$ and setting $f=1$ would lead to IR divergences. Nevertheless, our main formula \eqref{[QL,O1 ... On]} (or equivalently \eqref{[QL, TL]}) is essentially stating that the operator $[\textbf{Q}_L, -]$ differs from $[Q_0,-]$ exactly by $+ R_{L}(e_\otimes^{i F/ \hbar}; \frac{1}{2} (\hat{S}_0 +L,\hat{S}_0+L))$, and hence in the expressions in Theorem \ref{[QL, OL]} such terms are absent, i.e. those expressions are IR-finite.

\item If one, nevertheless, formally sets $f=1$ in $L$, in the second term in the right hand side of \eqref{[QL,O1 ... On]}, then $(\hat{S}_0 +L, \hat{S}_0+L)= (S,S)=0$, and formally 
$\frac{1}{2} R_{n,1}(\mathcal{O}_1 \otimes \dots \otimes \mathcal{O}_n; (\hat{S}_0 +L,\hat{S}_0 +L) \otimes e_\otimes^{iL/\hbar}) \rightarrow 0$. Therefore,
\begin{align}
[\textbf{Q}_L, T_{L,n}(\mathcal{O}_1\otimes \dots \otimes \mathcal{O}_n)] =[Q_0,  T_{L,n}(\mathcal{O}_1\otimes \dots \otimes \mathcal{O}_n)]_\star, \quad \text{\emph{(FORMALLY)}}.
\end{align}
That is, the commutator of the interacting BRST charge and interacting fields coincides with the commutator of the free BRST charge with them.
\end{enumerate}
\end{rmk}

\subsubsection*{Formal BRST-invariance of the S-matrix}

In a renormalization scheme with $A(e_\otimes^L)=0$, the Ward identity \eqref{WI} for $F=L$ expresses the commutator of the free BRST charge and the local S-matrix $\mathcal{S}(L)= T(e_\otimes^{i L/\hbar})$:
\beq
[Q_0, \mathcal{S}(L)] \eqos T\big( \frac{1}{2}(\hat{S}_0 +L, \hat{S}_0+L) \otimes e_\otimes^{i L/ \hbar}\big).
\eeq
As we explained in the second point of the Remark \ref{remark-Q}, if one formally sets $f=1$ in $L$, then $(\hat{S}_0 +L, \hat{S}_0+L)= (S,S)=0$ and $[Q_0, -]= [\textbf{Q}_L, -]$ . Therefore, this leads to the conclusion that if $A(e_\otimes^L)=0$, then the local S-matrix formally commutes with the quantum BRST charge:
\begin{align} \label{[QL,T]}
[\textbf{Q}_L, \mathcal{S}(L)] \eqos 0, \quad \text{(FORMALLY)}.
\end{align}


\subsection*{Interacting consistency conditions}

In order to prove the nilpotency of the derivation $[Q_L, -]$ on the interacting algebra, we need to obtain a consistency condition for the interacting anomalies \eqref{hat-A-n}. 
\begin{thm}\label{consistencycond}
In a renormalization scheme such that $A(e_\otimes^L)=0$, the interacting anomalies $A_{L,n}$ defined by \eqref{hat-A-n} satisfy the following \textbf{interacting consistency condition} for $F = \int f \wedge \mathcal{O}$ for \emph{all} $\mathcal{O} \in \textbf{P}^p(M)$, $f \in \Omega_0^{4-p}(M)$
\begin{align}\label{master-cons-cond}
(\hat{S} + F, A_{L}(e_\otimes^F)) +  A_L \big( \{ \frac{1}{2} (\hat{S} + F, \hat{S} + F) +  A_{L}(e_\otimes^F) \} \otimes e_\otimes^F \big) =0.
\end{align}
\end{thm}
\begin{proof}
To prove the above this identity, we use the free consistency condition \eqref{cons} which is valid for all bosonic functionals $F$. 
 We replace $F$ with $F+ \tau_1 \mathcal{O}_1 + \dots \tau_n \mathcal{O}_n$ in \eqref{cons}, differentiate with respect to $\tau_1 \dots \tau_n$ and set $\tau_i =0$. This procedure, together with the relations \eqref{tau-O-(1)}, \eqref{tau-O-(2)}, \eqref{tau-O-(3)}, in exactly the same way as in the proof of Theorem \ref{[QL, OL]}, leads to
\begin{align}\nonumber
& (\hat{S}_0 + F, A( \mathcal{O}_1 \otimes \dots \otimes \mathcal{O}_n \otimes e_\otimes^F)) \\ \nonumber
 &+ \sum_{i}(\mathcal{O}_{i}, {A}(\mathcal{O}_{1} \otimes \dots \otimes \mathcal{O}_{i-1} \otimes \mathcal{O}_{i+1} \otimes \dots \otimes \mathcal{O}_{n} \otimes e_\otimes^F )) \\ \nonumber
& +  \sum_{i=1}^n A(\mathcal{O}_1 \otimes \dots \otimes (\hat{S}_0 + F, \mathcal{O}_i) \otimes \dots \otimes \mathcal{O}_n \otimes e_\otimes^F ) \\ \nonumber
& +\sum_{I_2}^{n} A(  {(\mathcal{O}_i, \mathcal{O}_j)}_{i, j \in I_2} \otimes \bigotimes_{k \in I_2^c}  \mathcal{O}_k \otimes e_\otimes^F )\\ \nonumber
& +\sum_{I}  A \Big(  A(\bigotimes_{i \in I} \mathcal{O}_i \otimes e_\otimes^F ) \otimes \bigotimes_{j \in I^c}  \mathcal{O}_j \otimes e_\otimes^F  \Big)\\ \label{(S0+F, AO1On)}
 &+ A(\mathcal{O}_1 \otimes \dots \otimes \mathcal{O}_n \otimes \big( \frac{1}{2} (\hat{S}_0+F,\hat{S}_0 +F) +A(e_\otimes^F)  \big) \otimes e_\otimes^F) =0.
\end{align} 
Now setting $F={\hat{S}_\ia}$, and using $(\hat{S}_0 + {\hat{S}_\ia}, \hat{S}_0 + {\hat{S}_\ia})=0$ and $A(e_\otimes^{\hat{S}_\ia})=0$  in the above equation, we arrive at 
\begin{align} \nonumber
 \hat{\hs} {A}_{L, n}(  \mathcal{O}_1  \otimes \dots \otimes \mathcal{O}_n) & +\sum_{i}\big({A}_{L, n-1}(\mathcal{O}_{1} \otimes \dots \otimes \mathcal{O}_{i-1} \otimes \mathcal{O}_{i+1} \otimes \dots \otimes \mathcal{O}_{n}), \mathcal{O}_{i} \big) \\ \nonumber
  & +  \sum_{i=1}^n {A}_{L, n}(  \mathcal{O}_1  \otimes \dots \otimes \hat{\hs} \mathcal{O}_i \otimes \dots \otimes \mathcal{O}_n)\\ \nonumber
& + \sum_{I_2}^{n} {A}_{L, n-1}(  {(\mathcal{O}_i, \mathcal{O}_j)}_{i, j \in I_2} \otimes \bigotimes_{k \in I_2^c}  \mathcal{O}_k)\\ \label{qAO1On-bosonic}
&+\sum_{I}  {A}_{L, n - |I|+1} \Big(  {A}_{L, |I|}(\bigotimes_{i \in I} \mathcal{O}_i) \otimes \bigotimes_{j \in I^c}  \mathcal{O}_j \Big)=0 ,
\end{align}
which is the expanded identity \eqref{master-cons-cond} at order $n$.

\end{proof}
\subsection{Quantum BRST operator and quantum anti-bracket}

Let us now look at the expansion of the interacting anomalous Ward identity \eqref{int-AWI}, to lowest orders. It turns out that in the first and second orders, one can combine the operators $\hat{\hs}$ and $(-, -)$  of the classical theory with anomalies which are of quantum nature (i.e., vanish as $\hbar \rightarrow 0$).
\begin{dfn}\label{dfn:q-BRST-diff}
\begin{enumerate}
\item The \textbf{quantum BRST operator} $\hat{q}$, is the linear map
\beq
\hat{q}: \textbf{P}(M) \rightarrow  \textbf{P}(M)[[\hbar]],
\eeq
defined by
\beq\label{hat-q}
\hat{q} \cO := \hat{\hs}\cO + \hat{A}_1(\cO). 
\eeq
\item The \textbf{quantum anti-bracket} $(-, -)_\hbar$ is the bi-linear map
\beq
 (-,-)_\hbar:  \textbf{P}(M) \times  \textbf{P}(M) \rightarrow  \textbf{P}(M)[[\hbar]],
\eeq 
defined by 
\begin{equation}\label{q-anti-bracket}
(\cO_1, \cO_2)_\hbar := (\mathcal{O}_1,\mathcal{O}_2) + (-1)^{\eps_1} \hat{A}_2(\mathcal{O}_1\otimes \mathcal{O}_2),
\end{equation}
where $\eps_1$ is the Grassmann parity of $\cO_1$.
\end{enumerate}
\end{dfn}
Using the definitions of $\hat{q}$ and $(-, -)_\hbar$, we may equivalently write \eqref{int-AWI} for all bosonic fields $\cO_1 \dots \cO_n$ in the following form
\begin{align} 
\nonumber
  \frac{1}{i \hbar}[\textbf{Q}_L,T_{L,n}(\mathcal{O}_1 \otimes \dots  \otimes \mathcal{O}_n)] = &  \sum_{i=1}^n T_{L,n}(  \mathcal{O}_1  \otimes \dots \otimes \hat{q} \mathcal{O}_i \otimes \dots \otimes \mathcal{O}_n)\\ \nonumber
& + \frac{\hbar}{i} \sum_{1 \le i <j}^{n} T_{L,n-1}( \mathcal{O}_1 \otimes \dots \otimes (\mathcal{O}_i, \mathcal{O}_j)_\hbar  \otimes \dots \otimes  \mathcal{O}_n)\\  \label{[QL,O1On] -q-(-,-)}
&+\sum_{I} (\frac{\hbar}{i})^{|I|-1}  T_{L, n - |I|+1} \Big( \hat{A}_{|I|}(\bigotimes_{i \in I} \mathcal{O}_i) \otimes \bigotimes_{j \in I^c}  \mathcal{O}_j \Big),
\end{align}
where $|I|>2$.  The similar expression for $\cO_1 \dots \cO_n$ with Grassmann parity $\eps_1 \dots \eps_n$ is given in Appendix \ref{app1} equation \eqref{[QL,O1On]-fermionic}. 
In particular for $n=1,2$, \eqref{[QL,O1On]-fermionic} gives
\begin{align} \label{[QL,OL]}
[\textbf{Q}_L, \mathcal{O}_L(x)] &= i \hbar (\hat{q} \mathcal{O}(x))_L,\\ \nonumber
[\textbf{Q}_L, T_{L,2}(\mathcal{O}_1(x)\otimes \mathcal{O}_2(y))] &= i \hbar T_{L,2}\big( \{\hat{q} \mathcal{O}_1(x) \otimes \mathcal{O}_2(y)\} + (-1)^{\eps_1} \{ \mathcal{O}_1(x)\otimes \hat{q} \mathcal{O}_2(y) \}\big) \\ \label{QL,O1O2LL}
& \quad + (-1)^{\eps_1} \hbar^2 \big({(\mathcal{O}_1(x), \mathcal{O}_2(y))}_\hbar \big)_L.
\end{align}
Therefore, $(-, -)_\hbar$ may be interpreted as the failure of $\hat{q}$ to  be a derivation.
\begin{cor} In a renormalization scheme in which $A(e_\otimes^L)=0$, we have:
\begin{enumerate}
\item The quantum BRST operator \eqref{hat-q} is nilpotent
\begin{align} \label{q2=0}
 \hat{q}^2 =0,
\end{align}
\item The quantum BRST operator is compatible with the quantum anti-bracket \eqref{q-anti-bracket}
\begin{align} \label{q(O1,O2)}
\hat{q} ( \mathcal{O}_1 , \mathcal{O}_2 )_\hbar = (\hat{q}\mathcal{O}_1, \mathcal{O}_2 )_\hbar - (-1)^{\eps_1}(\mathcal{O}_1, \hat{q}\mathcal{O}_2 )_\hbar, \\ \nonumber
\end{align}
\item The quantum anti-bracket satisfies the following \textbf{quantum graded Jacobi identity}
 \begin{align} \nonumber
& (-1)^{(\eps_1+1) (\eps_3+1)} {(\mathcal{O}_1, {( \mathcal{O}_2 , \mathcal{O}_3 )}_\hbar)}_\hbar  + (-1)^{(\eps_2+1) (\eps_1+1)}  {(\mathcal{O}_2, {( \mathcal{O}_3 , \mathcal{O}_1 )}_\hbar)}_\hbar \\ \nonumber
& +  (-1)^{(\eps_3+1) (\eps_2+1)}   {(\mathcal{O}_3, {( \mathcal{O}_1 , \mathcal{O}_2 )}_\hbar)}_\hbar = (-1)^{\eps_1 + \eps_2 + \eps_3 + \eps_1 \eps_3}\Big\{q \hat{A}_3( \mathcal{O}_1 \otimes \mathcal{O}_2 \otimes \mathcal{O}_3) \\ \nonumber
& +   \hat{A}_3( q \mathcal{O}_1 \otimes \mathcal{O}_2 \otimes  \mathcal{O}_3)  + (-1)^{\eps_1} \hat{A}_3(\mathcal{O}_1 \otimes q \mathcal{O}_2 \otimes \mathcal{O}_3)\\ \nonumber 
&  + (-1)^{\eps_1 + \eps_2}\hat{A}_3(\mathcal{O}_1 \otimes \mathcal{O}_2 \otimes q \mathcal{O}_3) \Big\}.\\ \label{An-Jacobi-ferm}
\end{align}
\end{enumerate}
\end{cor}
\begin{proof}
To prove \eqref{q2=0}, we use the consistency condition \eqref{qAO1On-ferm} for anomalies which in the case of $n=1$ gives
\begin{align}
\hat{q} \hat{A}_1(\mathcal{O}) + \hat{A}_1 (\hat{\hs} \mathcal{O}) =0.
\end{align}
 and calculate
\begin{align*}
\hat{q}^2 \mathcal{O} &= \hat{q} \hat{\hs} \mathcal{O} +  \hat{q} \hat{A}_1( \mathcal{O}) \\ \nonumber
& = \hat{\hs}^2 \mathcal{O} +  \hat{A}_1(\hat{\hs} \mathcal{O}) -  \hat{A}_1(\hat{\hs} \mathcal{O}) \\
& = 0.
\end{align*}

To prove the identity \eqref{q(O1,O2)}, we need \eqref{qAO1On-ferm} for $n=2$ which reads
\begin{align}\nonumber
 \hat{\hs} \hat{A}_2(\mathcal{O}_1 \otimes \mathcal{O}_2) &- ( \mathcal{O}_1 , \hat{A}_1(\mathcal{O}_2) )  - (-1)^{\eps_1}  ( \hat{A}_1(\mathcal{O}_1) , \mathcal{O}_2 ) \\
&+ \hat{A}_2 (\hat{q} \mathcal{O}_1 \otimes \mathcal{O}_2) - (-1)^{\eps_1} \hat{A}_2 (\mathcal{O}_1 \otimes \hat{q} \mathcal{O}_2) + (-1)^{\eps_1} \hat{A}_1((\mathcal{O}_1 , \mathcal{O}_2)_\hbar) =0.
\end{align}
Now adding $(-1)^{\eps_1} \big( \hat{\hs} ( \mathcal{O}_1 , \mathcal{O}_2 ) - (\hat{\hs}\mathcal{O}_1, \mathcal{O}_2 ) + (-1)^{\eps_1}(\mathcal{O}_1, \hat{\hs}\mathcal{O}_2 ) =0 \big)$ to the above equation, we arrive at \eqref{q(O1,O2)}. 

To prove the quantum Jacobi identity \eqref{An-Jacobi-ferm}, consider \eqref{qAO1On-ferm} for the particular case of $n=3$:
\begin{align}\nonumber
&\hat{q} \hat{A}_3( \mathcal{O}_1 \otimes \mathcal{O}_2 \otimes \mathcal{O}_3) \\ \nonumber
 & + \hat{A}_3(\hat{q}\mathcal{O}_1 \otimes \mathcal{O}_2 \otimes \mathcal{O}_3) + (-1)^{\eps_1} \hat{A}_3(\mathcal{O}_1 \otimes \hat{q}\mathcal{O}_2 \otimes \mathcal{O}_3)  + (-1)^{\eps_1 + \eps_2}\hat{A}_3(\mathcal{O}_1 \otimes \mathcal{O}_2 \otimes \hat{q} \mathcal{O}_3) \\ \nonumber
 & + (\cO_1, \hat{A}_2 ( \mathcal{O}_2 \otimes \mathcal{O}_3)) \\ \nonumber
& + (-1)^{\eps_1 + \eps_2} \hat{A}_2(\cO_1 \otimes (\cO_2, \cO_3))  + (-1)^{\eps_1} \hat{A}_2(\cO_1 \otimes \hat{A}_2(\cO_2 \otimes \cO_3)) \\ \nonumber
& +  (-1)^{\eps_{1} \eps_2 + \eps_1 \eps_3 } \Big\{ (\cO_2, \hat{A}_2 ( \mathcal{O}_3 \otimes \mathcal{O}_1)) + (-1)^{ \eps_2 + \eps_3} \hat{A}_2(\cO_2 \otimes (\cO_3, \cO_1))\\ \nonumber
&  +(-1)^{ \eps_2}  \hat{A}_2(\cO_2 \otimes \hat{A}_2(\cO_3 \otimes \cO_1)) \Big\} + (-1)^{\eps_{1} \eps_3 + \eps_2 \eps_3 } \Big\{(\cO_3, \hat{A}_1 ( \mathcal{O}_2 \otimes \mathcal{O}_2)) \\ \nonumber
& + (-1)^{\eps_1 + \eps_3 }\hat{A}_2(\cO_3 \otimes (\cO_1, \cO_2)) + (-1)^{\eps_3 } \hat{A}_2(\cO_3 \otimes \hat{A}_2(\cO_1 \otimes \cO_2)) \Big\}  =0.
\end{align}
Adding the classical Jacobi identity
\beq \nonumber
(-1)^{\eps_2}{(\mathcal{O}_1, {( \mathcal{O}_2 , \mathcal{O}_3 )})}  + (-1)^{\eps_1\eps_2 + \eps_1 \eps_3 + \eps_3}  {(\mathcal{O}_2, {( \mathcal{O}_3 , \mathcal{O}_1 )})} + (-1)^{\eps_1\eps_3 + \eps_2 \eps_3 + \eps_1}  {(\mathcal{O}_3, {( \mathcal{O}_1 , \mathcal{O}_2 )})} =0
\eeq
 to the above and factoring out the sign factor $ (-1)^{\eps_1 + \eps_2 + \eps_3 + \eps_1 \eps_3}$, we obtain \eqref{An-Jacobi-ferm}.

\end{proof}

\begin{rmk}
\begin{enumerate}
\item By property \textbf{A10} of the anomaly, on basic fields and anti-fields, the quantum BRST operator coincides with the classical one, i.e.,
\begin{align}
\hat{q} \Phi = \hat{\hs} \Phi, \quad \hat{q} \Phi^\ddag = \hat{\hs} \Phi^\ddag.
\end{align}

\item The interacting consistency condition \eqref{qAO1On-ferm} for $n=1$ takes the form:
\beq
\hat{q} \hat{A}_1(\cO) + \hat{A}_1( \hat{\hs} \cO)=0.
\eeq 
Now, using that by property \textbf{A10} of the anomaly $\hat{A}_1(\Phi)=0 = \hat{A}_1(\Phi^\ddag)$, we obtain
\beq
 \hat{A}_1( \hat{\hs} \Phi)=0=\hat{A}_1( \hat{\hs} \Phi^\ddag).
\eeq

\item The identity \eqref{An-Jacobi-ferm} in the limit of $\hbar \rightarrow 0$ gives the (classical) graded Jacobi identity \eqref{class-Jacobi} of the classical anti-bracket. However, the quantum corrections prevent the quantum anti-bracket to satisfy the classical Jacobi identity. Similar violations have been observed in \cite{Alfaro-Damgaard} when analyzing the Hamiltonian Batalin-Vilkovisky formalism for the non-abelian group of field reparametrization transformations.
\end{enumerate}
\end{rmk}

\subsection{Nilpotency of $[\textbf{Q}_L, -]$} \label{sectionQ2=0}
 
We have so far derived how the quantum field $\textbf{Q}_L$, associated with the classical Noether charge $Q$ of the BV-BRST symmetry, acts on arbitrary quantum fields via the $\star$-commutator. It necessarily has to give the correct classical limit as $\hbar$ goes to zero. Indeed, since $\hat{A}_1(\mathcal{O}(x))$ is of order $O(\hbar)$, from the expression \eqref{[QL,OL]} it follows that
\begin{align}
\lim_{\hbar \rightarrow 0} \frac{i}{\hbar}[\textbf{Q}_L, - ] \eqos \hat{\hs}.
\end{align}
However, giving the correct classical limit is not a sufficient condition for $[\textbf{Q}_L, -]$ to define the action of the BRST symmetry on interacting quantum fields; in addition, it has to be nilpotent.

Using \eqref{q2=0}, it is now easy to see that $[\textbf{Q}_L, -]$, when acting on $\mathcal{O}_L$ is nilpotent, 
\begin{align}\label{[QL[QL-]]=0}
[\textbf{Q}_L, [\textbf{Q}_L, \mathcal{O}_L]] \eqos  (\hat{q}^2 \mathcal{O})_L =0.
\end{align}
One can then verify that it is also nilpotent when acting on the product of $n$ interacting fields. To see this, note that from the graded Jacobi identity \eqref{Jacobi-commutator} we obtain
\begin{align}\label{[Q,O1-star-On]-fermionic}
[\textbf{Q}_L, {\mathcal{O}_1}_L \star\dots \star {\mathcal{O}_n}_L] &\eqos \sum_{i} (-1)^{\sum_{i<k} \eps_l}  {\mathcal{O}_1}_L \star\dots \star (\hat{q}\mathcal{O}_{i})_L \star \dots  \star {\mathcal{O}_n}_L.
\end{align}
which gives \eqref{s0-Leibnizrule} in the free theory, i.e., in the limit where the coupling constant $\lambda$ is set to zero.
Now applying once again $[\textbf{Q}_L, -]$ on both sides of \eqref{[Q,O1-star-On]-fermionic} and using that $\hat{q}^2=0$, that $[\textbf{Q}_L, -]$ is a graded derivation and that the Grassmann parity of $\hat{q}\mathcal{O}_{i}$ are opposite to that of $\cO_i$, we find
\begin{align}\label{[QL[QL]]}
\big[\textbf{Q}_L, [\textbf{Q}_L, {\mathcal{O}_1}_L \star\dots \star {\mathcal{O}_n}_L] \big] \eqos 0.
\end{align}
We next show that $[\textbf{Q}_L, -]$ acting on $T_{L,n}(\mathcal{O}_1 \otimes \dots \otimes \mathcal{O}_n)$ is also nilpotent. 
\begin{thm} \label{QL2=0}
In a renormalization scheme such that $A(e_\otimes^L)=0$, we have  
\begin{align}
\big[ \textbf{Q}_L, [\textbf{Q}_L, T_L(e_\otimes^{iF/\hbar})] \big] \eqos 0.
\end{align}
 \end{thm}
\begin{proof}

We have already proven the statement for $n=1$. Before proving the claim for all $n$, let us explicitly verify it for $n=2$. Using $\hat{q}^2=0$ and \eqref{q(O1,O2)}, we have 
\begin{align}\nonumber
\frac{1}{(i \hbar)^2}[\textbf{Q}_L, [\textbf{Q}_L, T_{L,2}(\mathcal{O}_1 \otimes  \mathcal{O}_2)]]  & \eqos  (-1)^{\eps_1}  \frac{\hbar}{i} \big( (-1)^{\eps_1} (\mathcal{O}_1, \hat{q}\mathcal{O}_2 )_\hbar -   (\hat{q}\mathcal{O}_1, \mathcal{O}_2 )_\hbar +\hat{q} ( \mathcal{O}_1 , \mathcal{O}_2 )_\hbar \big)_L\\ \label{[QL[QL,T2]]}
& =  0.
\end{align} 
We prove the claim for all $n$, by applying $[\textbf{Q}_L, -]$ on both sides of the interacting anomalous Ward identity \eqref{int-AWI}. We obtain 
\begin{align} \nonumber
 & \big[ \textbf{Q}_L, [\textbf{Q}_L, T_L(e_\otimes^{iF/\hbar})] \big] \\ \nonumber
 & \eqos -\big[ \textbf{Q}_L, T_L \big( ( \frac{1}{2} (\hat{S} + F, \hat{S} + F) +  \hat{A}(e_\otimes^F) ) \otimes e_\otimes^{iF/\hbar} \big) \big] \\ \nonumber
 & = \frac{\hbar}{i} \frac{d}{d \tau} \big[ \textbf{Q}_L, T_L \big( e_\otimes^{i\big( F + \tau ( \frac{1}{2} (\hat{S} + F, \hat{S} + F) +  \hat{A}(e_\otimes^F) ) \big) /\hbar} \big) \big] \big|_{\tau =0} \\ \nonumber
 & =  T_L \big( \big((\hat{S} + F, ( \frac{1}{2} (\hat{S} + F, \hat{S} + F) +  \hat{A}(e_\otimes^F) ) \big) e_\otimes^{iF /\hbar} \big)\\ \nonumber
 & \quad +  T_L \big(  {A}_L \big( \big( \frac{1}{2} (\hat{S} + F, \hat{S} + F) +  \hat{A}(e_\otimes^F)  \big) \otimes e_\otimes^F \big) \otimes  e_\otimes^{iF /\hbar} \big)\\ \nonumber
 & \quad +  T_L \big(  \big( \frac{1}{2} (\hat{S} + F, \hat{S} + F) +  \hat{A}(e_\otimes^F) \big) \otimes \big( \frac{1}{2} (\hat{S} + F, \hat{S} + F) +  \hat{A}(e_\otimes^F) \big)  e_\otimes^{iF /\hbar} \big)\\ \nonumber
 & =  T_L \big(  (\hat{S} + F, \hat{A}(e_\otimes^F)) +  {A}_L \big( \big( \frac{1}{2} (\hat{S} + F, \hat{S} + F) +  \hat{A}(e_\otimes^F)  \big) \otimes e_\otimes^F \big) \otimes e_\otimes^{iF /\hbar} \big)\\ 
 & =0,
\end{align}
where in the forth line, we used that $(\hat{S} +F, (\hat{S} +F, \hat{S} +F))=0$ by Jacobi identity, the sixth line vanishes using the graded symmetry of $T_{L,n}$ and the fact that $\frac{1}{2} (\hat{S} + F, \hat{S} + F) +  \hat{A}(e_\otimes^F)$ is fermionic, and the last line vanishes by the consistency condition \eqref{master-cons-cond}  for ${A}_L(e_\otimes^F)$ which holds due to $A(e_\otimes^L)=0$. 
\end{proof}

\subsection{Quantum gauge invariant observables}\label{Q-g-inv-scheme-indep}

The algebra $\hat{\bF}_L = \hat{\bW}_L / \mathcal{J}_0$ that we constructed above does not correspond to the renormalized physical, gauge invariant observables of the Yang-Mills theory as it is a quantization of the classical enlarged theory defined by $\hat{S}$ and thus it includes gauge-variant and un-physical interacting fields, such as the vector potential and ghosts. We now want to define within $\hat{\bF}_L$ a subalgebra $\bF_L$ of gauge invariant observables.  
As discussed in Section \ref{Classical}, at the classical level, one can recover the gauge invariant observables as elements of the $\hat{\hs}$-cohomology class at ghost number zero. Motivated by this fact, we make the following definition.
\begin{dfn}
The \textbf{on-shell algebra of gauge invariant observables} $\bF_L \subset \hat{\bF}_L$ is defined by
\begin{align}\label{def-onshell-phys-alg}
\textbf{F}_L := \frac{\ker [\textbf{Q}_L, -] }{\Im [\textbf{Q}_L, -]}, \quad \text{at ghost number }0.
\end{align}
\end{dfn}

We now want to understand what interacting fields belong to the algebra of gauge invariant observables $\bF_L$, defined in \eqref{def-onshell-phys-alg}.

\begin{prop}\label{[Q, Psi]=0}
Let \emph{$\mathcal{O} \in \textbf{P}(M)$} be a classical gauge invariant operator, i.e., 
\beq
\hat{\hs} \mathcal{O}=0,
\eeq
with ghost number $0$. If the cohomology ring $H_1(\hat{\hs}, M)$ is trivial and there exists a renormalization scheme in which $A(e_\otimes^L)=0$, then there exists another scheme such that 
\begin{align}
\hat{q} \cO=0.
\end{align}
\end{prop}
\begin{proof}

We first note that since $\hat{\hs} \mathcal{O}=0$, we have $\hat{q}\cO = A_{L,1}(\mathcal{O})$.
We proceed by showing that the ``obstruction'' , $A_{L,1}(\cO)$, for $\hat{q} \cO$ to vanish, can be removed by passing to a new renormalization scheme.

Let us consider the expansion of the anomaly ${A}_{L,1}(\cO)= \sum_{n=1} \hbar^n {A}^n_{L,1}(\cO)$ in powers of $\hbar$. Then, the anomalies in different schemes turn out to be related by \cite{Hollands:2007zg}
\begin{align} \nn
\tilde{A}^m_{L,1}(\cO)&= {A}^m_{L,1}(\cO) + \hat{\hs} D^{m}(\mathcal{O} \otimes e_\otimes^L) - D^{m}(\hat{\hs} \mathcal{O} \otimes e_\otimes^L)+ A^k_{L,1}({D}^l_{L,1}(\cO) )  - D^k_{L,1}({A}^l_{L,1}(\cO) ), \\ \label{A-tilde-A-(m)}
\end{align}
where we have set $D^n(e_\otimes^L)$ for all $n=1,2, \dots.$, and where $l+k=m$. We now want to choose suitable finite counter terms ${D}^n_{L,1}(\cO)$ such that in the new scheme the anomaly $\tilde{A}^n_{L,1}(\cO)$ vanishes for all $n$. 

From the nilpotency of $\hat{q}$, we have for all $\cO \in \textbf{P}(M)$
\begin{align}\nn
0 &= \hat{q}^2 \cO \\ \nn
&= (\hat{\hs} + A_{L,1}(-))(\hat{\hs} \cO + A_{L,1}(\cO))\\ \label{qA(O)+A(sO)=0}
& = \hat{q} A_{L,1}(\cO) + A_{L,1}(\hat{\hs} \cO).
\end{align}
Therefore, for those $\cO$ with $\hat{\hs}\cO=0$ it follows that
\beq
\hat{q} A_{L,1}(\cO) =0.
\eeq
Now assume that ${A}^n_{L,1}(\cO)=0$ for all $n <m$. Then the above equation at order $\hbar^m$ gives
\begin{equation}\label{sAm=0}
\hat{\hs} {A}^m_{L,1}(\cO) =0.
\end{equation}
Since $\mathcal{O}$ has ghost number $0$, ${A}^m_{L,1}(\cO)$ has ghost number 1 and thus it belongs to $ H_1(\hat{\hs}, M)$ which is trivial by assumption. Therefore,
\begin{equation}\label{A(O)=sb}
{A}^m_{L,1}(\cO(x)) = \hat{\hs} b^m(x),
\end{equation}
for some $b^m(x) \in \textbf{P}_0(M).$ We now use this $b^m$ to redefine the time-ordered products by setting the following finite counter terms:
\begin{align}\label{D^m=-b^m}
D^{m}(\mathcal{O}(x) \otimes e_\otimes^L) = - b^m(x),
\end{align}
which from \eqref{A-tilde-A-(m)} results in
\begin{align}
\tilde{A}^m_{L,1}(\cO(x)) &= {A}^m_{L,1}(\cO) -  \hat{\hs}b^{m}(x) =0.
\end{align}
That is, in the new scheme the anomaly vanishes at order $\hbar^m$. 
Iterating the argument, we can fully remove the anomaly to all orders in $\hbar$.
\end{proof}

\begin{rmk}
For the case of the pure Yang-Mills theory, when $G$ is semi-simple with no abelian factors, $H(\hat{s}, M)$ is generated by elements of the form \cite{Barnich:2000zw}
\begin{equation}
{\displaystyle \prod_{k}} r_{t_k}(g, R, \nabla R, \dots, \nabla^k R) {\displaystyle \prod_{i}} p_{r_i}(C) {\displaystyle \prod_{j}} \Theta_{r_j}(F, \mathcal{D}F, \dots \mathcal{D}^l F)
\end{equation}
where $p_r$ and $\Theta_s$ are invariant polynomials of the Lie-algebra of $G$ and $r_t$ is a local functional of the metric $g$, the Riemann tensor $R$ and its derivatives. However, at ghost number $1$ the above expression vanishes as there is no invariant monomial $p_r$, and thus $H^p_1(\hat{s},M)$ is trivial. Therefore by the above Proposition \ref{[Q, Psi]=0}, $[\textbf{Q}_L, \mathcal{O}_L]=0$, for $\hat{s}\mathcal{O}=0$.
\end{rmk}

\subsection*{Hilbert space representation}

We have now collected all the tools which are required to represent the algebra of observables as linear operators with a dense, invariant domain on a Hilbert space $\mathcal{H}_L$. This space is constructed using a \emph{deformation process} \cite{dutsch1999local} from a Hilbert space $\mathcal{H}_0$ on which the free algebra
\begin{align}\label{F0}
\textbf{F}_0 := \frac{\Ker [Q_0, -] }{\Im [Q_0, -]}, \quad \text{at ghost number }0.
\end{align}
 is represented, as we review below. 

For theories without local gauge symmetry, such as the scalar field theory, given a quasi free, Hadamard state $\omega$ on the algebra $\bF_0$, via the celebrated GNS construction, one obtains a representation $\pi_\omega: \bF_0 \rightarrow \text{End}(\mathcal{H}_\omega)$ of the algebra of observables as linear operators with the so-called \emph{microlocal domain of smoothness} \cite{Brunetti:1999jn} $\mathcal{D}_\omega$ on a Hilbert space $\mathcal{H}_\omega$. In theories with local gauge symmetry, the perturbative quantization of the classically gauge-fixed theory, led to the free on-shell algebra $\hat{\bF}_0$ \eqref{hat-F-0} which contains gauge-variant and unphysical elements. In fact, $\hat{\bF}_0$ can only be represented on an indefinite inner product space.

In order to obtain a positive definite inner product, one in addition has to impose a \emph{positivity condition} \cite{dutsch1999local} on the representation, as we describe below. 

\begin{dfn}\label{def:pos-cond}
Let $\omega$ be a  quasi free, Hadamard state on the algebra ${\bF}_0$, and let $\pi^\omega_0: {\bF}_0 \rightarrow \text{End}(\mathcal{K}^\omega_0)$ be a faithful representation of $\bF_0$ on a space $ (\mathcal{K}^\omega_0, \langle -, - \rangle)$ with indefinite inner product and such that all the anti-fields are represented by the $0$ operator, and let $\mathcal{D}_\omega \subset \mathcal{H}_\omega$ the dense and invariant microlocal domain of smoothness \cite{Brunetti:1999jn}.
Furthermore, let $Q_0 \in \hat{\bF}_0$ be the free BRST charge \eqref{:Q0:}. This representation is called to satisfy the \textbf{positivity conditions}, if
\begin{align}\label{positivity-cond-1}
&\langle \phi, \phi \rangle \ge 0, \quad \forall \phi \in \ker {\pi}^\omega_0(Q_0) \cap \mathcal{D}_\omega,\\ \label{positivity-cond-2}
 &\phi \in \ker {\pi}^\omega_0(Q_0) \cap \mathcal{D}_\omega \text{ and } \langle \phi, \phi \rangle=0 \implies \phi \in \text{Im } {\pi}^\omega_0(Q_0) \cap \mathcal{D}_\omega.
\end{align}
\end{dfn}

\begin{thm}[\cite{dutsch1999local}] \label{Hilbert-H_0}
Let $\pi^\omega_0$ be a representation of the algebra ${\bF}_0$ \eqref{F0} as in Definition \ref{def:pos-cond} which for all $\phi, \psi \in \mathcal{D}_\omega$ and for all $ \cO \in \bF_0$ satisfies
\begin{align}
\langle {\pi}^\omega_0(\cO^\dag) \phi, \psi \rangle &= \langle \phi, {\pi}^\omega_0(\cO) \psi \rangle, \\
\langle {\pi}^\omega_0(Q_0) \phi, \psi \rangle &= \langle \phi, {\pi}^\omega_0(Q_0) \psi \rangle, \\
{\pi}^\omega_0(Q_0)^2& =0, \quad \text{on } \mathcal{D}_\omega,
\end{align}
and the positivity conditions of Definition \ref{def:pos-cond}. Then 
\beq
\mathcal{H}_0 := \frac{\ker {\pi}^\omega_0(Q_0)}{\Im {\pi}^\omega_0(Q_0)}
\eeq
is a pre Hilbert space.
\end{thm}

It is shown in \cite{dutsch1999local}, Section 4.3,  that the construction of $\mathcal{H}_0$ is ``stable under deformations''. This means that once the positivity conditions \eqref{positivity-cond-1} and \eqref{positivity-cond-2} for the representation $\pi_0$ are satisfied and the interacting BRST charge $Q_L$ is nilpotent, then the algebra $\bF_L $ can be represented on a Hilbert space $\mathcal{H}_L$ with a positive-definite\footnote{In that reference, a formal power series $A=\sum_{n} \lambda^n A_n \in \mathbb{C}[[\lambda]]$ is called positive if there exists another formal power series $B=\sum_{n} \lambda^n B_n$, such that $B^* B =A$, i.e. $A_n = \sum_{k=0}^n \bar{B}_k B_{n-k}$.} inner product which is induced from the inner product on $\mathcal{H}_0$.

We have shown in Section \ref{WardIdentities} that the free BRST charge is nilpotent. As a corollary of the nilpotency of $[Q_L, -]$, we now show that the interacting BRST charge is also nilpotent which is required for the algebra $\bF_L $ of interacting gauge invariant observables to admit a Hilbert space representation.

\begin{cor}[\textbf{of Theorem} \ref{QL2=0}]\label{cor-Q2=0} If $A(e_\otimes^L)=0$, the quantum BRST charge is nilpotent modulo an element in $\mathcal{J}_0$, i.e. 
\begin{align}
Q_L^2 \eqos 0.
\end{align}
 
\end{cor}
\begin{proof}
Since $Q_L$ is odd, we have 
\begin{align}
[Q_L, Q_L] \eqos 2 Q_L \star Q_L = 2 Q_L^2.
\end{align}
From the nilpotency of $[Q_L, -] $ and the graded Jacobi identity,  we have for all $T_{L,n}(\mathcal{O}_1 \otimes \dots \otimes \cO_n) \in \hat{\bF}_L$
\begin{align} \nn
0 & \eqos [Q_L, [Q_L, T_{L,n}(\mathcal{O}_1 \otimes \dots \otimes \cO_n)]] \\ \nn
& = \frac{1}{2} \big[[Q_L, Q_L], T_{L,n}(\mathcal{O}_1 \otimes \dots \otimes \cO_n)\big]\\
  &= [Q_L^2, T_{L,n}(\mathcal{O}_1 \otimes \dots \otimes \cO_n)].
\end{align}
By Proposition 2.1 in \cite{Hollands:2001nf}, it follows that $Q_L^2$ must be a multiple of the identity element \textbf{1} in $\hat{\bF}_L$
\begin{align}
Q_L^2 \eqos  k \textbf{1},
\end{align}
for some constant scalar $k$ with ghost number 2 made out of background fields. However, there is no such a constant in the theory, thus $k=0$. 
\end{proof}
 
Therefore, the only requirements for a positive-definite Hilbert space representation of $\bF_L$ which must be fulfilled are the positivity conditions. 
In \cite{Hollands:2007zg}, it is shown that the positivity condition is satisfied in the following setting. Let $U \subset \Sigma$ be an open domain in a Cauchy surface $\Sigma$, with a smooth boundary $\partial U$ and compact closure and vanishing first deRahm cohomology $H^1(U, d)$ \footnote{These conditions on $U$ are imposed in order to exclude the existence of ``zero-modes''.} and let $D(U) \subset M$ be its domain of dependence. Then, within $D(U)$, Hadamard two point functions $\omega$ and $\omega^{\mu \nu}$ satisfying the consistency relations \eqref{omega-consistency} are explicitly constructed.

\section{Comparison with other approaches}\label{comparison}

The local and covariant approach to gauge theories in curved space-times that we developed in the body of the thesis, differs from the other approaches in the literature in many respects that are explained in the text. In this chapter, we compare our approach with three different formulations of gauge theories and outline a number of (formal) similarities between those approaches and ours. 

\subsection{Path integral formalism}\label{BVPI}
There are obvious differences between our approach and the Path integral approach. In the latter, one defines the generating functional for the correlation functions of the Euclidean theory via an integral over the infinite dimensional manifold of all field configurations. Such a path integral is a priori only formal in many respects: the measure on the infinite dimensional space does not exists, and if one wants to make sense of it as a formal power series in the coupling constant each individual term suffers from both IR and UV divergences. Even ignoring these difficulties, the path integral is a state-dependent quantity as it generates correlation functions in a specific (globally defined) state. This makes it difficult to appreciate the local and covariant nature of the renormalization ambiguities \citep{Hollands:2001fb} in arbitrary curved space-times with no preferred state. 
Despite those differences, there are formal similarities between the two approaches in studying gauge theories. They both lead to a kind of algebraic structure which is called the Batalin-Vilkovisky (BV) algebra. Let us briefly outline the types of arguments in the path integral formalism which result in the emergence of the BV-structure, and then compare that with our formalism.

\subsubsection*{Batalin-Vilkovisky formalism}
In the path integral approach to gauge theories (see e.g. \cite{henneaux1992quantization}, \cite{Weinberg:1996kr}), one modifies the ``measure'' by higher order terms in $\hbar$ to obtain a ``gauge invariant measure $\mathcal{D} \phi$''. This corrections can be equivalently seen as quantum correction to the classical action $S$ and classical operators $\mathcal{O}$, i.e.
\begin{equation} \label{<O>}
\langle \mathcal{O} \rangle = \int D \phi \mathcal{O}_\hbar e^{- i S_\hbar/\hbar},
\end{equation}
where $\mathcal{O}_\hbar = \mathcal{O} + \hbar (\mathcal{O}_\hbar)_1 + \hbar^2 (\mathcal{O}_\hbar)_2 + \dots$ and where $S_\hbar= S + \hbar (S_\hbar)_1 + \hbar^2 (S_\hbar)_2 + \dots$ is called the \emph{quantum action}. The precise form of $W$ is determined be requiring it to be a solution to the quantum master equation (QME):
\begin{equation}
(S_\hbar, S_\hbar)= - 2 i \hbar \Delta S_\hbar,
\end{equation}
which ensures that $\langle \mathcal{O} \rangle$ is independent of the gauge-fixing. In the above equation, $\Delta= (-1)^{\eps}\frac{\delta}{\delta \Phi \delta \Phi^\ddag}$ is the so-called \emph{BV laplacian}. Using $\Delta$, one defines the \emph{quantum differential} 
\begin{align}\label{sigma}
\sigma = (S_\hbar, -) - i \hbar \Delta,
\end{align}
which satisfies the following properties:
\begin{itemize}
\item[(i)] $\sigma^2=0,$
\item[(ii)] $\sigma (\mathcal{O}_1 \mathcal{O}_2) = (\sigma \mathcal{O}_1) \mathcal{O}_2 + \mathcal{O}_1 (\sigma \mathcal{O}_2) + (\mathcal{O}_1 , \mathcal{O}_2)$,
\item[(iii)] $\sigma(\mathcal{O}_1, \mathcal{O}_2)= (\sigma \mathcal{O}_1, \mathcal{O}_2) + (\mathcal{O}_1, \sigma \mathcal{O}_2)$.
\end{itemize}
Once the quantum Master equation is satisfied, one can prove the following \emph{Ward identity for correlation functions} 
\begin{equation}\label{SDsigmapsi}
\sum_{i=1}^n \langle \mathcal{O}_1 \dots \sigma \mathcal{O}_i \dots \mathcal{O}_n \rangle = 0,
\end{equation}
which for $n=1$ is reduced to 
\beq\label{<sigmaO>=0}
\langle \sigma \mathcal{O}\rangle =0.
\eeq 

We now point out the following analogies and differences between the path integral approach and ours.

\begin{itemize}
\item[(1)] The definitions of our quantum BRST differentials $\hat{q}$ defined in equation \eqref{hat-q} and $\sigma$ given in equation \eqref{sigma} are similar in that they both are given by classical BRST $\hat{\hs}$ plus higher $\hbar$ corrections. However, they are different in that the quantum corrections for $\sigma$ is given by $\hbar\big(((S_\hbar)_1, -) + i \Delta) + \sum_{n \ge 2}\hbar^n ((S_\hbar)_n, -)$ which are ill-defined since $\Delta$ is a singular operator and $W_n$ are in general IR divergent, whereas the quantum corrections for $\hat{q}$ are given by $ {A}_L( - )$ defined in equation \eqref{hat-A-n} which is a well-defined local operator. Indeed, as first noted by authors of \cite{Fredenhagen:2011mq} (however, see section \ref{BVpAQFT} below), ${A}_L( -)$ may be seen as the ``renormalized BV laplacian''.

\item[(2)] The properties (i), (ii), (iii) of $\sigma$ define the BV algebra. Evidently, (i) and (iii) are similar to properties \eqref{q2=0} and \eqref{q(O1,O2)} of $\hat{q}$, and property (ii) is similar to \eqref{QL,O1O2}.  The difference between them is the presence of quantum anti-bracket $(-, -)_\hbar$ which differs from the classical anti-bracket by terms of order $O(\hbar^2)$ (which are given by \eqref{q-anti-bracket}). Therefore, one may see our BV data \eqref{q2=0}, \eqref{q(O1,O2)}, and \eqref{QL,O1O2} as defining the ``renormalized BV algebra''.

\item[(3)] The QME is in general violated by potential anomalies, and as it turns out in the path integral approach, such anomalies belong to the same cohomological class as $A(e_\otimes^L)$ (i.e. $H_1^4(\hat{\hs}| d, M)$). Therefore, our proof that $A(e_\otimes^L)=0$, and hence the Ward identity \eqref{WI} holds, may be taken as the counterpart for the proof that the QME holds in the path integral framework. 

\item[(4)] As we have proved in the present work, from our Ward identity \eqref{WI} it ultimately follows that $[{Q}_L, -]$ is a nilpotent derivation and hence one can define the algebra of physical observables as the cohomology of $[{Q}_L, -]$. Of course observables in the image of $\hat{q}$ are quotiented out and their expectation value in a physical state $|\Psi \rangle \in \mathcal{H}_{L}$ (see Section \ref{Q-g-inv-scheme-indep}) vanishes:
\beq
\langle \hat{q} \mathcal{O} \rangle_\Psi =0.
\eeq 
This fact is clearly comparable with \eqref{<sigmaO>=0} which states that the expectation value of observables in the image of $\sigma$ vanishes if the quantum action satisfies the QME. For $n$ operator insertions, we obtain from \eqref{[QL,O1On] -q-(-,-)}
\begin{align}\nn 
& \sum_{i=1}^n T \big\langle  \mathcal{O}_1 \dots \hat{q} \mathcal{O}_i \dots \mathcal{O}_n  \big \rangle_\Psi   +  \sum_{1 \le i <j}^{n} T\langle \mathcal{O}_1 \dots (\mathcal{O}_i, \mathcal{O}_j)_\hbar  \dots   \mathcal{O}_n   \rangle_\Psi \\ \label{<WI>}
& +  \sum_{I} (\frac{\hbar}{i})^{|I|-1}  T \big\langle  {A}_{L,|I|}(\bigotimes_{i \in I} \mathcal{O}_i) \otimes \prod_{j \in I^c}  \mathcal{O}_j  \big\rangle_\Psi =0, 
\end{align}
where $T\langle  \mathcal{O}_1 \dots  \mathcal{O}_n  \rangle_\Psi := \langle  \Psi | T_{L, n}(\mathcal{O}_1\otimes \dots \otimes \mathcal{O}_n) | \Psi \rangle $ are renormalized time-ordered $n$ point functions of the theory in the state $|\Psi \rangle \in \mathcal{H}_{L}$. Comparison with \eqref{SDsigmapsi} reveals that the identity \eqref{<WI>} involves quantum corrections to  \eqref{SDsigmapsi} and may be interpreted as the  \emph{renormalized Ward identity for correlation functions}. 

\end{itemize}

\subsection{Batalin-Vilkovisky formalism in the pAQFT approach}\label{BVpAQFT}

The closest approach to ours is the BV formalism in the framework of perturbative algebraic quantum field theory (pAQFT) developed in \cite{Fredenhagen:2011mq}. While this approach is in the same sprit as ours, there are still notable differences which we point out here.

\begin{itemize}
\item[(1)] In the pAQFT approach, contrary to \eqref{hatS0+I},  one makes a different split of the action $\hat{S}$ into free $\tilde{S}_0$ and interacting $\tilde{S}_\ia$ parts by putting all terms depending on the anti-fields into $\tilde{S}_\ia$. Although this does not affect the classical BRST differential, i.e. $\hat{\hs} = (\tilde{S}_0 + \tilde{S}_1, -) = (S_0 + S_1, -)$, the free action $\tilde{S}_0$ only acts on anti-fields, i.e. $(\tilde{S}_0, \Phi^\ddag(x))= \frac{\delta \tilde{S}_0}{\delta \Phi(x)}$ and $(\tilde{S}_0, \Phi)=0$.
One then formulates a similar anomalous Ward identity in the form
\begin{equation}\label{AnWIBV}
(\tilde{S}_0, T(e_\otimes^{iF/ \hbar})) = \frac{i}{\hbar}\left(T( \frac{1}{2}(\tilde{S}_0 + F, \tilde{S}_0 + F)\otimes e_\otimes^{iF/ \hbar}) - i\hbar T( \Delta(F)\otimes e_\otimes^{iF/ \hbar}) \right),
\end{equation}
where $\Delta(F)$ defines the anomaly. Despite the obvious similarity to \eqref{AnWI}, a key difference is that, the left hand side of the above identity vanishes on-shell, contrary to \eqref{AnWI}. To clarify the issue let us elaborate on the proof of our anomalous Ward identity  given in \cite{Hollands:2007zg}. One first decomposes $\hat{\hs}= s + \sigma$ where $s$ is the BRST differential which only acts on fields, and $\sigma$ is the Koszul-Tate differential which acts only on anti-fields. The anomalous Ward identity is then obtained by adding two different identities: (1) an identity for $ s_0 T(e_\otimes^{iF/ \hbar})$ which gives an anomaly term $\delta(e_\otimes^F)$ and (2) an identity, originally derived in \cite{Brennecke:2007uj}, for $ \sigma_0 T(e_\otimes^{iF/ \hbar})$ which gives an anomaly term $\Delta(e_\otimes^F)$ . One then defines $A(e_\otimes^F)= \delta(e_\otimes^F) + \Delta(e_\otimes^F)$ and obtains \eqref{AnWI}. It seems that the identity \eqref{AnWIBV} is the second identity mentioned above which realizes the free Koszul-Tate differential on the local S-matrix.

\item[(2)] The quantum BV operator $\tilde{s}$ in the pAQFT approach is defined by
\begin{equation}\label{tildes}
 \tilde{s}  := R_{V}^{-1} \circ  (\tilde{S}_0, - ) \circ R_V,
\end{equation}
where $R_V = T(e_\otimes^{i V/ \hbar})^{-1} \star T(e_\otimes^{i V/ \hbar} \otimes -)$ is viewed here as an operator (the quantum M{\"o}ller operator) which takes a functional $\mathcal{O}$ and gives $\mathcal{O}_V$.  This definition differs from $\hat{q}$ for the following reason. From the nilpotency of $(\tilde{S}_0, -)$, it follows that $\tilde{s}^2=0$, which means that $\tilde{s} $ is always nilpotent by construction, irrespective of the presence or absence of an anomaly. Evidently, this is different from our quantum BRST differential $\hat{q}$ which is nilpotent only if $A(e_\otimes^L)=0$. 
Nevertheless, using the QME, one can show that $\tilde{s}$ takes the following form: 
\begin{equation}
\tilde{s} \mathcal{O} := \hat{\hs}\mathcal{O} - \Delta_V(\mathcal{O}), 
\end{equation}
which is analogous to our definition of $\hat{q}$, except for the difference between $\Delta_V(\mathcal{O})$ and ${A}_{L,1}(\mathcal{O})$ which was explained in point 1.

\item[(3)] In \cite{Rejzner:2013ak}, it is shown that the quantum BV operator, on-shell, can be written as the commutator with an interacting charge $Q$, i.e.
\begin{align}\label{[RVO, RVQ]}
[R_V(Q), R_V(\mathcal{O})] = i \hbar R_V(\tilde{s} \mathcal{O}), \quad \text{mod } \mathcal{J}_0. 
\end{align}
As pointed out above, on-shell, where the above formula is valid, the right hand side vanishes. Therefore, the formula \eqref{[RVO, RVQ]} seems to express that the charge $R_V(Q)$ commutes with \emph{all} interacting fields on-shell. 
This is of course a plausible statement for $Q$ being the generator of the Koszul-Tate differential. However, this is obviously different from identity \eqref{[QL,O]} which expresses that the interacting BRST charge ${Q}_L$ only commutes with interacting fields $\mathcal{O}_L$ for which $\hat{q}\mathcal{O}=0$.
Consequently, being in the cohomology of $[R_V(Q), -]$ does not seem to provide a criterion for selecting the physical observables and for selecting the physical states of the theory in a Hilbert space representation. 

\end{itemize}

\subsection{Renormalization group flow equation framework and Ward identities}
The renormalization group Flow equation framework \cite{Polchinski:1983gv, Wetterich:1992yh, Muller:2002he, Kopper:2005jq} is a mathematically rigorous framework to the renormalization of quantum fields in flat Euclidean field theories. The application of this approach to gauge theories is worked out in \cite{frob2016all} where a proof of perturbative renormalizability is given in the sense that all correlation functions of arbitrary composite local operators fulfil suitable Ward identities. 

In \cite{frob2016all}, similar operators to our $\hat{q}$ and $(-, -)_\hbar$ appears  in analysing the gauge invariance of the Euclidean theory. These are given by $\hat{q}^{\text{E}} \cO= \hat{\hs} \cO + \hat{A}_1(\cO)$ and $(\cO_1, \cO_2)^{\text{E}}_\hbar = (\cO_1, \cO_2) + \hat{A}_2(\cO_1 \otimes \cO_2)$. Here $\hat{A}_1$ and $\hat{A}_2$ are of order $O(\hbar)$ and supported on the diagonal (contact terms) and hence are analogous to our ${A}_{L,1}$ and ${A}_{L,2}$. 

However, one major difference with our approach is that in the flow equation framework no analogue of our anomalies with more than two insertions appear at all in the approach of \cite{frob2016all}. This fact has two consequences. First, contrary to \eqref{An-Jacobi-ferm}, the quantum anti-bracket satisfies the usual (classical) Jacobi identity \eqref{An-Jacobi-ferm} without ${A}_{L,3}$ terms. The second consequence can be seen by looking at the Ward identity which expresses gauge invariance in this framework. This is an identity for the vacuum expectation values $\langle  \mathcal{O}_1 \dots \mathcal{O}_n\rangle^{\text{E}}_0$ of $n$ operators $\mathcal{O}_1, \dots, \mathcal{O}_n$ (Schwinger functions), and takes the form
\begin{align}
 \sum_{i=1}^n \langle  \mathcal{O}_1 \dots \hat{q}^{\text{E}} \mathcal{O}_i \dots \mathcal{O}_n\rangle^{\text{E}}_0 +  \sum_{1 \le i <j}^{n} \langle \mathcal{O}_1 \dots (\mathcal{O}_i, \mathcal{O}_j)^{\text{E}}_\hbar  \dots   \mathcal{O}_n\rangle^{\text{E}}_0= 0.
\end{align}
This is obviously similar to the identity \eqref{<WI>} with the difference that the terms containing ${A}_{L,n}$ with $n \ge 3$ are absent.

This difference might be a consequence of different renormalization conditions that are imposed in the two approaches. In fact in our approach, besides the specific renormalization scheme that we choose in which the anomaly is absent and the BRST current is conserved, we do not impose any further renormalization condition. However, in the flow equation approach, in deriving the Ward identities one imposes specific boundary conditions for the flow equation which amounts to choosing a specific renormalization scheme. In this respect, our approach seems more general in not restricting to a specific renormalization scheme. 
Nevertheless, if ${A}_{L,n}$ with $n \ge 3$ can be made not to appear at all by a choice of renormalization condition in one approach, presumably one has to be able to pass to a renormalization scheme in our approach in which ${A}_{L,n}=0$ for $n \ge 3$. However, to our knowledge, this does not seem to be possible. 

\section{Outlook}

In this paper, we have developed new algebraic structures in quantum gauge theories which enables one to construct the algebra of renormalized gauge-invariant observables in a model-independent fashion. Such structures, namely quantum BRST differential $\hat{q}=\hat{s} + O(\hbar)$ and quantum anti-bracket $(-,-)_\hbar = (-,-) + O(\hbar)$ are indeed analogous to the classical ones modified with certain quantum corrections. The new structures seem to provide sufficient tools to investigate further open issues in gauge theories such as the issue of Gauge-fixing independence as we explain in the following.

In section \ref{Classical}, we pointed out that for perturbative quantization of gauge theories one has to choose a particular way to fix the gauge in order to render the equations of motion hyperbolic. The natural question is, then, whether and in which sense different quantum field theories defined with different (in general, non-linear) gauge-fixings are equivalent?

To be more specific, different gauge-fixings may arise, for instance, from a family gauge-fixing fermions $\psi(\xi)$, for $\xi \in \mathbb{R}$ with
\begin{equation}
\psi(\xi)= \int_M \bar{C}_I(\nabla^\mu A_\mu^I + \frac{\xi}{2} B^I),
\end{equation}
which gives rise to the family of linear covariant gauges. $\xi=1$ corresponds to the Feynman gauge (which was considered in this work) and the limit $\xi \rightarrow 0$ corresponds to the Landau gauge. In this case, the question of the equivalence of quantum field theories defined with $\psi(\xi)$ and $\psi(\xi')$ may be stated as follows. At the classical level, there exists an isomorphism $\mathcal{O} \mapsto e^{(-,\delta\psi)} \mathcal{O}$, with $\delta \psi = \psi(\xi') - \psi(\xi)$ between the cohomologies of the BRST differentials $\hat{s}_\xi$ and $\hat{s}_{\xi'}$ which ensures that the observables of the two theories are in one-to-one correspondence.  Based on the analogy between classical and quantum structures worked out in the paper, one can then formulate \cite{Taslimi-gauge-ind} the gauge-fixing independence at the quantum level as the existence of an isomorphism $\mathcal{O}_{L(\xi)} \mapsto e^{(-,\delta \psi)_\hbar} \mathcal{O}_{L(\xi')}$ between $\hat{q}_\xi$ and $\hat{q}_{\xi'}$ cohomologies.

\section*{Acknowledgement}
This work is part of the author's PhD dissertation. I gratefully acknowledge financial
support by the Max Planck Institute for Mathematics in the Sciences and its International
Max Planck Research School (IMPRS). I am very grateful to my supervisor Stefan Hollands for suggesting me the subject and for his guidance and many fruitful discussions. I would like to thank Jochen Zahn for a careful reading of an earlier version of this work, many stimulating discussions, important suggestions for improving the work and pointing out an error in the first preprint version. I would also like to thank Markus Fr{\"o}b for many fruitful discussions specially for bringing to my knowledge the definition of the quantum anti-bracket in the flow equation framework. Furthermore, discussions with Kasia Rejzner about the BV formalism is gratefully acknowledged.    

\appendix
\section{Graded symmetries and derivations}\label{app1}
 In this appendix, we derive the identities \eqref{AnWI}, \eqref{int-AWI} and \eqref{master-cons-cond} when $F_i$ are either bosonic or fermionic. The starting point for that is the anomalous Ward identity. However, the anomalous Ward identity, as stated in Theorem \ref{AnWIthrm}, only applies to bosonic $F$'s, as for fermionic ones $(\hat{S}_0 +F, \hat{S}_0 + F) =(\hat{S}_0 , F) + (F,F) + (F, \hat{S}_0)=0$. Let $\eps_i$ be the Grassmann parity of $F_i$. This means $\eps_i = 0 $ mod $2$ if $F_i$ is bosonic (even), and $\eps_i = 1$ mod $2$ if $F_i$ is fermionic (odd). It turns out that the anomalous Ward identity for local functionals $F_i$ with Grassmann parity $\eps_i$ takes the form
\begin{align} \nonumber
\hat{s}_0   T_n(F_1 \otimes \dots \otimes F_n) &= \sum_{k=0} (-1)^{\sum_{l <k} \eps_l } T_n(F_1 \otimes \dots \otimes \hat{s}_0 F_k \otimes \dots \otimes F_n)\\ \nonumber
& + \frac{\hbar}{i} \sum_{I_2}  (-1)^{\eps_{I_2}+ \eps_i} T_{n-1}  \big( {(F_i, F_{j})}_{i,j \in I_2} \otimes    \bigotimes_{k \in I_2^c}  F_k \big)\\ \label{AnWI-order-n-fermionic}
& + \sum_{I} (\frac{\hbar}{i})^{|I|-1} (-1)^{\eps_{I}} T_{n-|I|+1} \big(A_{|I|}(\bigotimes_{i \in I} F_i) \otimes \bigotimes_{j \in I^c}  F_j \big),
\end{align}
where $I$ is a non-empty and ordered partition of the set $\{ 1, 2, \dots, n \}$ and $I^c$ is the complement partition and $|I_2|=2$, and  where $\eps_{I_2}$ and $\eps_I$ are sings that are obtained by reordering $F_i$'s into $F_1  \dots  F_n$, i.e. 
\begin{equation}
\prod_{j \in I_2} F_j \prod_{k \in I_2^c} F_k = (-1)^{\eps_{I_2} } F_1 \dots F_n,
\end{equation} 
\begin{equation}
\prod_{j \in I} F_j \prod_{k \in I^c} F_k = (-1)^{\eps_{I}} F_1 \dots F_n.
\end{equation} 
The above identity is derived by (1) starting from \eqref{AnWI} for bosonic $g_1 F_1, \dots ,g_n F_n$, where $g_i$ are anti-commuting numbers, $g_i F_j = (-1)^{\eps_i \eps_j} F_j g_i$, with the same Grassmann parity as $F_i$, and (2) using the following identities 
 \begin{align}\label{(g1O1,g2O2)}
(g_1 F_1, g_2 F_2) &= (-1)^{(\eps_1 +1)\eps_2} g_1 g_2 (F_1, F_2),\\ \label{An(g1O1-gnOn)}
A_n(g_1 F_1 \otimes \dots \otimes g_n F_n) & = (-1)^{\sum_i \eps_i + \sum_{i<j} \eps_i \eps_j} g_1 \dots g_n A_n( F_1 \otimes  \dots \otimes F_n),
\end{align}
which are consequences of the symmetry property \eqref{graded-symm} and the graded symmetry of $A_n$.

Similarly the interacting anomalous Ward identity \eqref{int-AWI} for all \emph{$F_i \in \textbf{P}(M)$} with Grassmann parity $\eps_i$ takes the form
\begin{align} 
\nonumber
  \frac{1}{i \hbar}[{Q}_L,T_{L,n}(F_1 \otimes \dots  \otimes F_n)] = &  \sum_{k=0} (-1)^{\sum_{l <k} \eps_l } T_{L,n}(  F_1  \otimes \dots \otimes \hat{\hs} F_i \otimes \dots \otimes F_n)\\ \nonumber
& + \frac{\hbar}{i} \sum_{I_2}  (-1)^{\eps_{I_2}+ \eps_i} T_{L,n-1}(  {(F_i, F_j)}_{i, j \in I_2} \otimes \bigotimes_{k \in I_2^c}  F_k)\\  \label{[QL,O1On]-fermionic}
&+ \sum_{I} (\frac{\hbar}{i})^{|I|-1} (-1)^{\eps_{I}} T_{L, n - |I|+1} \Big(  {A}_{L,|I|}(\bigotimes_{i \in I} F_i) \otimes \bigotimes_{j \in I^c}  F_j \Big).
\end{align}

Finally, the interacting consistency condition \eqref{master-cons-cond} for all \emph{$F_i \in \textbf{P}(M)$} with Grassmann parity $\eps_i$ takes the form
\begin{align} \nonumber
 \hat{\hs} A_{L,n}&(  F_1  \otimes \dots \otimes F_n) \\ \nn
& +\sum_{i} (-1)^{\sum_{j} \eps_j + \eps_i(1+ \eps_{i+1} + \dots + \eps_n)} \big( A_{L,n-1}(F_{1} \otimes \dots \otimes F_{i-1} \otimes F_{i+1} \otimes \dots \otimes F_{n}), F_{i} \big) \\ \nonumber
  & +  \sum_{i=0} (-1)^{\sum_{l <i} \eps_l } A_{L,n}(  F_1  \otimes \dots \otimes \hat{\hs} F_i \otimes \dots \otimes F_n)\\ \nonumber
& + \sum_{I_2}^{n} (-1)^{\eps_{I_2}+ \eps_i} A_{L,n-1}(  {(F_i, F_j)}_{i, j \in I_2} \otimes \bigotimes_{k \in I_2^c}  F_k)\\  \nonumber
&+\sum_{I}  {A}_{L,n - |I|+1} (-1)^{\eps_{I}} \Big(  {A}_{L,|I|}(\bigotimes_{i \in I} F_i) \otimes \bigotimes_{j \in I^c}  F_j \Big)=0,\\ \label{qAO1On-ferm}
\end{align}
where the sum runs over all non-empty subsets $I$ of the set $\{ 1, 2, \dots, n \}$, $I^c$ is the complement subset and $|I_2|=2$.

\bibliographystyle{JHEP}
\bibliography{qBRSTcharge}

\providecommand{\href}[2]{#2}\begingroup\raggedright\begin{thebibliography}{10}

\bibitem{Hollands:2007zg}
S.~Hollands, \emph{{Renormalized Quantum Yang-Mills Fields in Curved
  Spacetime}}, \href{http://dx.doi.org/10.1142/S0129055X08003420}{\emph{Rev.
  Math. Phys.} {\bf 20} (2008) 1033--1172}.

\bibitem{Taslimi-ABJM}
M.~Taslimi~Tehrani, \emph{{Self-consistency of conformally coupled ABJM theory
  at the quantum level}},
  \href{http://dx.doi.org/10.1007/JHEP11(2017)153}{\emph{Journal of High Energy
  Physics} {\bf 2017} (Nov, 2017) 153}.

\bibitem{deMedeiros:2013mca}
P.~de~Medeiros and S.~Hollands, \emph{{Superconformal quantum field theory in
  curved spacetime}},
  \href{http://dx.doi.org/10.1088/0264-9381/30/17/175015}{\emph{Class. Quant.
  Grav.} {\bf 30} (2013) 175015}.

\bibitem{barnich1995general}
G.~Barnich, F.~Brandt and M.~Henneaux, \emph{General solution of the
  wess-zumino consistency condition for einstein gravity}, {\emph{Physical
  Review D} {\bf 51} (1995) R1435}.

\bibitem{Brunetti2016}
R.~Brunetti, K.~Fredenhagen and K.~Rejzner, \emph{Quantum gravity from the
  point of view of locally covariant quantum field theory},
  \href{http://dx.doi.org/10.1007/s00220-016-2676-x}{\emph{Communications in
  Mathematical Physics} {\bf 345} (2016) 741--779}.

\bibitem{Kugo:1979gm}
T.~Kugo and I.~Ojima, \emph{{Local Covariant Operator Formalism of Nonabelian
  Gauge Theories and Quark Confinement Problem}},
  \href{http://dx.doi.org/10.1143/PTPS.66.1}{\emph{Prog. Theor. Phys. Suppl.}
  {\bf 66} (1979) 1--130}.

\bibitem{Curci1976}
G.~Curci and R.~Ferrari, \emph{An alternative approach to the proof of
  unitarity for gauge theories},
  \href{http://dx.doi.org/10.1007/BF02730284}{\emph{Il Nuovo Cimento A
  (1965-1970)} {\bf 35} (Oct, 1976) 273--279}.

\bibitem{henneaux1992quantization}
M.~Henneaux and C.~Teitelboim, \emph{Quantization of gauge systems}.
\newblock Princeton university press, 1992.

\bibitem{Zinn-Justin}
J.~Zinn-Justin, \emph{Renormalization of gauge theories},
  \href{https://www.springer.com/de/book/9783540071600}{\emph{Bonn lectures
  1974, published in Trends in Elementary Particle Physics, Lecture Notes in
  Physics 37 pages 1-39, H. Rollnik and K. Dietz eds., Springer Verlag,
  Berlin.} (1975) }.

\bibitem{Brunetti:1995rf}
R.~Brunetti, K.~Fredenhagen and M.~Kohler, \emph{{The Microlocal spectrum
  condition and Wick polynomials of free fields on curved space-times}},
  \href{http://dx.doi.org/10.1007/BF02099626}{\emph{Commun. Math. Phys.} {\bf
  180} (1996) 633--652}.

\bibitem{Brunetti:1999jn}
R.~Brunetti and K.~Fredenhagen, \emph{{Microlocal analysis and interacting
  quantum field theories: Renormalization on physical backgrounds}},
  \href{http://dx.doi.org/10.1007/s002200050004}{\emph{Commun. Math. Phys.}
  {\bf 208} (2000) 623--661}.

\bibitem{Hollands:2001nf}
S.~Hollands and R.~M. Wald, \emph{{Local Wick polynomials and time ordered
  products of quantum fields in curved space-time}},
  \href{http://dx.doi.org/10.1007/s002200100540}{\emph{Commun. Math. Phys.}
  {\bf 223} (2001) 289--326}.

\bibitem{Hollands:2001fb}
S.~Hollands and R.~M. Wald, \emph{{Existence of local covariant time ordered
  products of quantum fields in curved space-time}},
  \href{http://dx.doi.org/10.1007/s00220-002-0719-y}{\emph{Commun. Math. Phys.}
  {\bf 231} (2002) 309--345}.

\bibitem{Hollands:2002ux}
S.~Hollands and R.~M. Wald, \emph{{On the renormalization group in curved
  space-time}},
  \href{http://dx.doi.org/10.1007/s00220-003-0837-1}{\emph{Commun. Math. Phys.}
  {\bf 237} (2003) 123--160}.

\bibitem{Hollands:2014eia}
S.~Hollands and R.~M. Wald, \emph{{Quantum fields in curved spacetime}},
  \href{http://dx.doi.org/10.1016/j.physrep.2015.02.001}{\emph{Phys. Rept.}
  {\bf 574} (2015) 1--35}.

\bibitem{dutsch1999local}
M.~D{\"u}tsch and K.~Fredenhagen, \emph{{A local (perturbative) construction of
  observables in gauge theories: the example of QED}}, {\emph{Communications in
  mathematical physics} {\bf 203} (1999) 71--105}.

\bibitem{Rejzner:2011au}
K.~Rejzner, \emph{{Fermionic fields in the functional approach to classical
  field theory}}, \href{http://dx.doi.org/10.1142/S0129055X11004503}{\emph{Rev.
  Math. Phys.} {\bf 23} (2011) 1009--1033}.

\bibitem{Iyer:1994ys}
V.~Iyer and R.~M. Wald, \emph{{Some properties of Noether charge and a proposal
  for dynamical black hole entropy}},
  \href{http://dx.doi.org/10.1103/PhysRevD.50.846}{\emph{Phys. Rev.} {\bf D50}
  (1994) 846--864}.

\bibitem{wald1990identically}
R.~M. Wald, \emph{On identically closed forms locally constructed from a
  field}, {\emph{Journal of mathematical physics} {\bf 31} (1990) 2378--2384}.

\bibitem{Barnich:2000zw}
G.~Barnich, F.~Brandt and M.~Henneaux, \emph{{Local BRST cohomology in gauge
  theories}},
  \href{http://dx.doi.org/10.1016/S0370-1573(00)00049-1}{\emph{Phys. Rept.}
  {\bf 338} (2000) 439--569}.

\bibitem{epstein1973role}
H.~Epstein and V.~Glaser, \emph{The role of locality in perturbation theory},
  in \emph{Annales de l'IHP Physique th{\'e}orique}, vol.~19, pp.~211--295,
  1973.

\bibitem{Brunetti:2001dx}
R.~Brunetti, K.~Fredenhagen and R.~Verch, \emph{{The Generally covariant
  locality principle: A New paradigm for local quantum field theory}},
  \href{http://dx.doi.org/10.1007/s00220-003-0815-7}{\emph{Commun. Math. Phys.}
  {\bf 237} (2003) 31--68}.

\bibitem{Radzikowski:1996pa}
M.~J. Radzikowski, \emph{{Micro-local approach to the Hadamard condition in
  quantum field theory on curved space-time}},
  \href{http://dx.doi.org/10.1007/BF02100096}{\emph{Commun. Math. Phys.} {\bf
  179} (1996) 529--553}.

\bibitem{dewitt1960radiation}
B.~S. DeWitt and R.~W. Brehme, \emph{Radiation damping in a gravitational
  field}, {\emph{Annals of Physics} {\bf 9} (1960) 220--259}.

\bibitem{Khavkine2016}
I.~Khavkine and V.~Moretti, \emph{{Analytic Dependence is an Unnecessary
  Requirement in Renormalization of Locally Covariant QFT}},
  {\emph{Communications in Mathematical Physics} {\bf 344} (2016) 581--620}.

\bibitem{DeWitt:2004xz}
B.~DeWitt and C.~DeWitt-Morette, \emph{{From the Peierls bracket to the Feynman
  functional integral}},
  \href{http://dx.doi.org/10.1016/j.aop.2004.07.005}{\emph{Annals Phys.} {\bf
  314} (2004) 448--463}.

\bibitem{Kallen:1950uha}
G.~Kallen, \emph{{Formal Integration of the Equations of Quantum Theory in the
  Heisenberg Representation}},
  \href{http://dx.doi.org/10.1007/978-3-319-00627-7_89}{\emph{Ark. Fys.} {\bf
  2} (1950) 371--385}. [,465(1950)].

\bibitem{Frob:2018buw}
M.~B. Fr{\"o}b, \emph{{Anomalies in time-ordered products and applications to
  the BV-BRST formulation of quantum gauge theories}},
  \href{https://arxiv.org/abs/1803.10235}{\emph{arXiv:1803.10235 [math-ph]}
  (2018) }.

\bibitem{Tehrani-Zahn}
M.~Taslimi~Tehrani and J.~Zahn, \emph{{Background independence in gauge
  theories}}, \href{https://arxiv.org/abs/1804.07640}{\emph{arXiv:1804.07640
  [math-ph]} (2018) }.

\bibitem{Alfaro-Damgaard}
J.~Alfaro and P.~H. Damgaard, \emph{Non-abelian antibrackets},
  \href{http://dx.doi.org/https://doi.org/10.1016/0370-2693(95)01533-7}{\emph{Physics
  Letters B} {\bf 369} (1996) 289 -- 294}.

\bibitem{Weinberg:1996kr}
S.~Weinberg, \emph{{The quantum theory of fields. Vol. 2: Modern
  applications}}.
\newblock \href{http://www.worldcat.org/search?q=isbn:9781139632478,
  9780521670548, 9780521550024}{, Cambridge University Press, 2013, }.

\bibitem{Fredenhagen:2011mq}
K.~Fredenhagen and K.~Rejzner, \emph{{Batalin-Vilkovisky formalism in
  perturbative algebraic quantum field theory}},
  \href{http://dx.doi.org/10.1007/s00220-012-1601-1}{\emph{Commun. Math. Phys.}
  {\bf 317} (2013) 697--725}.

\bibitem{Brennecke:2007uj}
F.~Brennecke and M.~Dutsch, \emph{{Removal of violations of the Master Ward
  Identity in perturbative QFT}},
  \href{http://dx.doi.org/10.1142/S0129055X08003237}{\emph{Rev. Math. Phys.}
  {\bf 20} (2008) 119--172}.

\bibitem{Rejzner:2013ak}
K.~Rejzner, \emph{{Remarks on Local Symmetry Invariance in Perturbative
  Algebraic Quantum Field Theory}},
  \href{http://dx.doi.org/10.1007/s00023-014-0312-x}{\emph{Annales Henri
  Poincare} {\bf 16} (2015) 205--238}.

\bibitem{Polchinski:1983gv}
J.~Polchinski, \emph{{Renormalization and Effective Lagrangians}},
  \href{http://dx.doi.org/10.1016/0550-3213(84)90287-6}{\emph{Nucl. Phys.} {\bf
  B231} (1984) 269--295}.

\bibitem{Wetterich:1992yh}
C.~Wetterich, \emph{{Exact evolution equation for the effective potential}},
  \href{http://dx.doi.org/10.1016/0370-2693(93)90726-X}{\emph{Phys. Lett.} {\bf
  B301} (1993) 90--94}.

\bibitem{Muller:2002he}
V.~F. Muller, \emph{{Perturbative renormalization by flow equations}},
  \href{http://dx.doi.org/10.1142/S0129055X03001692}{\emph{Rev. Math. Phys.}
  {\bf 15} (2003) 491}.

\bibitem{Kopper:2005jq}
C.~Kopper, \emph{{Renormalization theory based on flow equations}},
  \href{https://arxiv.org/abs/hep-th/0508143}{\emph{Prog. Math.} {\bf 251}
  (2007) 161--174}.

\bibitem{frob2016all}
M.~B. Fr{\"o}b, J.~Holland and S.~Hollands, \emph{All-order bounds for
  correlation functions of gauge-invariant operators in yang-mills theory},
  \href{http://dx.doi.org/10.1063/1.4967747}{\emph{Journal of Mathematical
  Physics} {\bf 57} (2016) 122301}.

\bibitem{Taslimi-gauge-ind}
M.~Taslimi~Tehrani, \emph{{Gauge-fixing independence in gauge theories in
  curved space-time}}, {\emph{In preparation} }.

\end{thebibliography}\endgroup
\end{document}